\documentclass[journal,draftcls,onecolumn,12pt,twoside]{IEEEtranTCOM}

\usepackage{epsfig}
\usepackage{graphicx}
\usepackage{epstopdf}
\usepackage{amsmath}
\usepackage{amsfonts}
\usepackage{amssymb}
\usepackage{setspace}
\usepackage{pgf,pgfarrows,pgfnodes,pgfautomata,pgfheaps,pgfshade}
\usepackage{tikz}
\usepackage{amsthm}

\theoremstyle{plain}
\newtheorem{thm}{Theorem}[section]
\newtheorem{lem}[thm]{Lemma}

\newtheorem{defn}{Definition}[section]
\theoremstyle{remark}
\newtheorem*{rem}{Remark}

\renewcommand{\thesection}{\arabic{section}}
\allowdisplaybreaks

\usepackage{amssymb}
\begin{document}
\title{Enhancing Secrecy Rate Region for Recent Messages for a Slotted Multiple  Access Wiretap Channel to Shannon Capacity Region}

\author{Shahid M Shah,~\IEEEmembership{Student~Member,~IEEE,}
	and~Vinod~Sharma,~\IEEEmembership{Senior~Member,~IEEE}% <-this % stops a space
\thanks{Part of the paper was presented in IEEE Wireless Communication and Networking Conference (WCNC), March 2015, New Orleans, LA, USA.}
\thanks{Shahid M Shah and Vinod Sharma are with Electrical communication Department, Indian Institute of Science, Bangalore, India.}}
\maketitle

%% Title, authors and addresses

%% use the tnoteref command within \title for footnotes;
%% use the tnotetext command for theassociated footnote;
%% use the fnref command within \author or \address for footnotes;
%% use the fntext command for theassociated footnote;
%% use the corref command within \author for corresponding author footnotes;
%% use the cortext command for theassociated footnote;
%% use the ead command for the email address,
%% and the form \ead[url] for the home page:
%% \title{Title\tnoteref{label1}}
%% \tnotetext[label1]{}
%% \author{Name\corref{cor1}\fnref{label2}}
%% \ead{email address}
%% \ead[url]{home page}
%% \fntext[label2]{}
%% \cortext[cor1]{}
%% \address{Address\fnref{label3}}
%% \fntext[label3]{}

%% use optional labels to link authors explicitly to addresses:
%% \author[label1,label2]{}
%% \address[label1]{}
%% \address[label2]{}

\begin{abstract}
Security constraint results in \textit{rate-loss} in wiretap channels. In this paper we propose a coding scheme for two user Multiple Access Channel with Wiretap (MAC-WT), where previous messages are used as a key to enhance the secrecy rates of both the users until we achieve the usual capacity region of a Multiple Access Channel (MAC) without the wiretapper (Shannon capacity region). With this scheme  all the messages transmitted in the recent past are secure with respect to all the information of the eavesdropper till now. To achieve this goal we introduce secret key buffers at both the users, as well as at the legitimate receiver (Bob). Finally we consider a fading MAC-WT and show that with this coding/decoding scheme we can achieve the capacity region of a fading MAC (in ergodic sense).
\end{abstract}


%% \linenumbers

%% main text
Wyner in his seminal paper \cite{wyner1975wire} on a degraded wiretap channel proved that by assigning multiple codewords to a single message, we can achieve reliability as well as security in a point to point channel. He characterized secrecy capacity for this channel. After a couple of decades of this work when wireless revolution began, researchers started extending Wyner's coding scheme (wiretap coding) in different directions. A single user fading wiretap channel was studied in \cite{gopala2008secrecy}, \cite{bloch2008wireless}. A secret key buffer was used in \cite{gungor2013secrecy} to mitigate the fluctuations in the secrecy capacity due to variations in the channel gain with time.

A multiple access channel with security constraints was studied in \cite{liang2008multiple} and \cite{tekin2008gaussian}. In \cite{liang2008multiple} the transmitting users treat each other as eavesdroppers and an achievable secrecy rate region is characterized. In some special cases the secrecy capacity region is also found. In \cite{tekin2008gaussian}  the authors consider the eavesdropper to be \textit{listening} at the receiving end. The authors provide an achievable secrecy-rate region. The secrecy-capacity region is not known for such a MAC. The same authors also studied a fading MAC with full channel state information (CSI) of Eve known at the transmitters. In \cite{shah2012achievable} this work is extended to the case when the CSI of Eve is not known at the transmitters. For a detailed review on information theoretic security, see \cite{liang2009information}, \cite{bloch2011physical}, and \cite{liu2010securing}.

In all these works a notion of weak secrecy was used, i.e., if $W$ is the message transmitted and Eve receives $Z^n$ for a codeword of length $n$ channel uses, then $I(W;Z^n)/n \rightarrow 0$, as $n\rightarrow \infty$. This notion of secrecy is not  stringent enough in various cases \cite{bloch2011physical}. Maurer in \cite{maurer1993secret} proposed a notion of \textit{strong secrecy}: $I(W;Z^n) \rightarrow 0$ as $n\rightarrow \infty$. For a point to point channel, he showed that it can be achieved without any change in secrecy capacity. Since then other methods have been proposed for  achieving strong secrecy \cite{devetak2005private}, \cite{csiszar2011information} and \cite{bloch2011strong}. The methods of \cite{devetak2005private} and \cite{bloch2011strong}  have been used to obtain strong secrecy for a MAC-WT in \cite{wiese2013strong} and \cite{yassaee2010multiple} respectively.

In all these works we observe that security is achieved at the cost of transmission rate. For a single user AWGN wiretap channel if $C_b$ is the capacity of the legitimate receiver (Bob) and $C_e$ is the capacity of Eve's channel, then the secrecy capacity of this channel is $C_s=(C_b-C_e)^+$, where $(x)^+=\max(0,x)$  (\cite{leung1978gaussian}). In recent years some work has been done to mitigate the secrecy-rate loss. Feedback channel is used in \cite{ardestanizadeh2009wiretap} and \cite{lai2008wiretap} to enhance the secrecy rate, and under certain conditions the authors prove that the secrecy capacity can approach the main channel capacity. In \cite{kang2010wiretap} the authors assume that the transmitter (Alice) and Bob have access to a secret key, and then they propose a coding scheme which utilizes that key to enhance the secrecy rate. Secure Multiplex scheme has been proposed in \cite{kobayashi2013secure} which achieves Shannon channel capacity for a point to point wiretap channel. In this model multiple messages are transmitted. The authors show that the mutual information of the currently transmitted message with respect to (w.r.t.) all the information received by Eve goes to zero as the codeword length $n\rightarrow \infty$.

Shah et al.  in \cite{shah2013previous} propose a simple coding scheme, without any feedback channel or access to some key, and enhance the secrecy capacity of a wiretap channel to the Shannon capacity of the main channel. In this work also, only the message currently being transmitted is secure w.r.t. all the information possessed by Eve. 
In \cite{shah2014achieving_MAC} we extended the coding scheme of \cite{shah2013previous} to a multiple access wiretap channel and showed that we can achieve Shannon capacity region of the MAC as the secrecy rate region, while keeping currently transmitted message secure w.r.t. all the information of Eve.
In this paper we extend the coding/decoding schemes of \cite{shah2013previous} and \cite{shah2014achieving_MAC} to a multiple access wiretap channel and prove that we can achieve Shannon capacity region of the MAC as the secrecy-rate region while keeping all recent messages secure w.r.t. the information possessed by Eve till present. Finally we achieve the same for a fading MAC-WT.

 Rest of the paper is organised as follows. In Section \ref{section_channel_model} we define the channel model and recall some previous results which will be used in this paper. We extend our coding/decoding scheme of \cite{shah2013previous} to two user discrete memoryless MAC-WT (DM-MAC-WT) in Section \ref{section_enhancing_rate} and prove the achievability of Shannon capacity region, under the security constraint that only the currently transmitted message is secured w.r.t. all the data received by Eve. In Section \ref{section_mac_with_buffer} we consider a two user DM-MAC-WT where each user, receiver, as well as Eve have infinite length buffers to store previous messages. We propose a coding scheme  to enhance the secrecy-rate region to Shannon capacity region of the usual MAC, this time with security constraint that \textit{all recent} messages are secure w.r.t. all the information possessed by Eve. In Section \ref{section_fading_mac} we consider a two user fading MAC-WT and extend the coding scheme of previous sections to enhance the secrecy-rate region of the fading MAC-WT to the Shannon Capacity region of the MAC in the ergodic sense. Section \ref{section_conclusion} concludes the paper. The Appendix at the end contains several lemmas used in the proofs of the main theorems.

\par
In this paper random variables will be denoted by capital letters $X,Y,Z$ etc., vectors will be denoted with upperbar letters, e.g., $\overline{X}=(X_1,\ldots, X_n)$, scalar constants will be denoted by lower case letters $a,b$ etc.
\section{Multiple Access Wiretap Channel}
\label{section_channel_model}
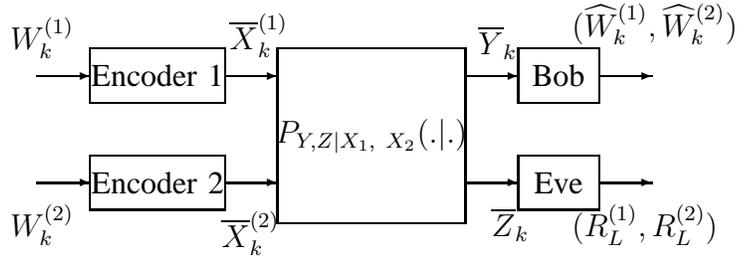
\begin{figure}[h]
	\setlength{\unitlength}{0.14in} % selecting unit length
	\centering % used for centering Figure
	\begin{picture}(32,15) % picture environment with the size (dimensions)
	% 32 length units wide, and 15 units high.
	\put(2,8.5){\framebox(5,2){Encoder 1}}
	\put(2,4.5){\framebox(5,2){Encoder 2}}
	\put(9,4){\framebox(7,6.5){$P_{Y,Z\lvert X_1,~X_2}(.\lvert .)$}}
	\put(18,8.5){\framebox(3,2){Bob}}
	\put(18,4.5){\framebox(3,2){Eve}}
	\put(0,9.5){\vector(1,0){2}}\put(7,9.5){\vector(1,0){2}}
	\put(7,5.5){\vector(1,0){2}}
	\put(0,5.5){\vector(1,0){2}}
	\put(16,9.5){\vector(1,0){2}}\put(16,5.5){\vector(1,0){2}}
	\put(21,9.5){\vector(1,0){2}}
	\put(21,5.5){\vector(1,0){2}}
	\put(-1,10.5) {$W^{(1)}_k$}
	\put(-1,3.5) {$W^{(2)}_k$}
	\put(20,11) {$(\widehat{W}^{(1)}_k,\widehat{W}^{(2)}_k)$}
	\put(20,3.5) {$(R_L^{(1)},R_L^{(2)})$}
	\put(7.2,10.5)
	{$\overline{X}^{(1)}_k$} \put(6.9,3.1){$\overline{X}^{(2)}_k$}\put(16.5,10.3) {$\overline{Y}_k$}
	\put(17,3.5){$\overline{Z}_k$}
	\end{picture}
	\caption{Discrete Memoryless Multiple Access Wiretap Channel} % title of the Figure
	\label{fig:DM_MAC} % label to refer figure in text
\end{figure}

A discrete memoryless multiple access channel with a wiretapper and two users is considered (Fig. \ref{fig:DM_MAC}). The channel is represented by transition probability matrix $p(y,z\lvert x_1,x_2)$ where $x_i \in \mathcal{X}_i$, is the channel input from user $i$, $i=1,2$, $y\in\mathcal{Y}$ is the channel output to Bob and $z\in \mathcal{Z}$ is the channel output to Eve. The sets $\mathcal{X}_1, \mathcal{X}_2, \mathcal{Y}, \mathcal{Z}$ are finite. The two users want to send messages $W^{(1)}$ and $W^{(2)}$ to Bob reliably, while keeping Eve ignorant about the messages.

 \begin{defn}
 	For a MAC-WT, a $(2^{nR_1},2^{nR_2},n)$ codebook consists of (1) message sets $\mathcal{W}^{(1)}$ and $\mathcal{W}^{(2)}$ of cardinality $2^{nR_1}$ and $2^{nR_2}$, (2) messages $W^{(1)}$ and $W^{(2)}$, which are uniformly distributed over the corresponding message sets $\mathcal{W}^{(1)}$ and $\mathcal{W}^{(2)}$ and are independent of each other, (3) two stochastic encoders,
 \end{defn} 
\begin{equation}
f_i:\mathcal{W}^{(i)} \rightarrow \mathcal{X}_i^n,~~i=1,2,
\end{equation}
and (4) a decoder at Bob, 
\begin{equation}
g:\mathcal{Y}^n\rightarrow \mathcal{W}^{(1)} \times \mathcal{W}^{(2)}.
\end{equation}
 The decoded messages are denoted by $(\widehat{W}^{(1)},\widehat{W}^{(2)})$.

The average probability of error at Bob is 
\begin{equation}
P_e^{(n)}\triangleq P\left\lbrace \left(\widehat{W}^{(1)},\widehat{W}^{(2)}\right) \neq \left(W^{(1)},W^{(2)}\right) \right\rbrace,
\end{equation}
and leakage rate at Eve is 
\begin{equation}
R_L^{(n)}=\frac{1}{n}I(W^{(1)},W^{(2)};Z^n).  \label{collective_leakage_def}
\end{equation}
\textit{Leakage Rate}: In \cite{tekin2008gaussian} the authors have defined two types of security requirements depending upon the trust of the transmitting users on each other. If each user is conservative such that when the other user is transmitting then it may compromise with Eve and provide Eve with its codeword, then  \textit{individual leakage} constraints
\begin{equation}
R_{L,1}^{(n)}=\frac{1}{n}I(W^{(1)};Z^n\lvert \overline{X}^{(2)}),  \label{leakage_user_1}
\end{equation}
\begin{equation}
R_{L,2}^{(n)}=\frac{1}{n}I(W^{(2)};Z^n\lvert \overline{X}^{(1)}),   \label{leakage_user_2}
\end{equation}
are relevant, where $\overline{X}^{(i)}$ denotes the codeword for user $i$.

In a scenario where users trust each other, \textit{collective leakage} 
\begin{equation}
R_L^{(n)}=\frac{1}{n}I(W^{(1)},W^{(2)};Z^n).  \label{collective_leakage}
\end{equation}
is relevant.
Since, $W^{(1)} \perp W^{(2)}$ and hence also $\overline{X}^{(1)} \perp \overline{X}^{(2)}$ where $X\perp Y$ denotes that random variable $X$ is independent of $Y$,
\begin{align}
nR_L^{(n)}&=I(W^{(1)},W^{(2)};Z^n)  \nonumber \\
&=I(W^{(1)};Z^n)+I(W^{(2)};Z^n\lvert W^{(1)}) \nonumber \\
&=H(W^{(1)})-H(W^{(1)}\lvert Z^n)+H(W^{(2)})-H(W^{(2)}\lvert Z^n,W^{(1)}) \nonumber \\
&\leq H(W^{(1)}\lvert X_2^n)-H(W^{(1)}\lvert Z^n,X_2^n)+H(W^{(2)}\lvert X_1^n) -H(W^{(2)}\lvert Z^n,X_1^n) \nonumber \\
&=I(W^{(1)};Z^n\lvert X_2^n)+I(W^{(2)};Z^n\lvert X_1^n) \nonumber \\
&=nR_{L,1}^{(n)}+nR_{L,2}^{(n)}   \label{leakage_2}
\end{align}
and hence, if individual leakage rates are small then so is the collective leakage rate. In this paper we consider the secrecy notion (\ref{collective_leakage}).
% % % % % % % % % % % % % % % % % % % % % %A % % % % % % % % % % % % % % % % % %AAAAAAAAAAAAAAAA % % % % % % % %

% % %AAAAAAAAAAAA % % % % % % % % % % % % % % % % % %

\begin{defn}
	 The secrecy-rates $(R_1,R_2)$ are achievable if there exists a sequence of codes $(2^{nR_1},2^{nR_2},n)$ with $P^{(n)}_e\rightarrow 0$ as $n\rightarrow \infty$ and
\begin{align}
\limsup_{n\rightarrow \infty} R_{L,i}^{(n)} =0,\quad \text{for}~i=1,2.
\end{align}
The secrecy-capacity region is the closure of the convex hull of achievable secrecy-rate pairs $(R_1,R_2)$.
\end{defn}

%\begin{figure}[!htb]
%\centering
%\includegraphics[width=0.45\textwidth]{MAC_model.pdf}
%\caption{Multiple Access Wiretap Channel}
%\label{fig:digraph}
%\end{figure}

In \cite{tekin2008gaussian}, a coding scheme to obtain the following rate region was proposed.

\begin{thm}
	 Rates $(R_1,R_2)$ are achievable with $\limsup_{n\rightarrow \infty}R_{L,i}^{(n)}=0,~i=1,2$, if there exist independent random variables $(X_1,X_2)$ as channel inputs satisfying
\begin{align}
R_1&< I(X_1;Y|X_2)-I(X_1;Z),  \nonumber \\
R_2&< I(X_2;Y|X_1)-I(X_2;Z),  \nonumber \\
R_1+R_2&< I(X_1,X_2;Y)-I(X_1;Z)-I(X_2;Z),  \label{MAC_WT_Secrecy_region}
\end{align}
where $Y$ and $Z$ are the corresponding symbols received by Bob and Eve.$~~ \square$
\end{thm}
\quad The secrecy capacity region for a MAC-WT is not known. If the secrecy constraint is not there then the capacity region for a MAC is obtained from the convex closure of the regions in Theorem 1 without the terms $I(X_i;Z),~i=1,2$ on the right side of (\ref{MAC_WT_Secrecy_region}) (Fig.\ref{MAC_region_compare}) \cite{ahlswede1973multi}. In the next section we show that we can attain the capacity region of a MAC even when some secrecy constraints are satisfied.

\begin{figure}[htb]
\centering
\begin{tikzpicture}
	\draw[arrows=->](0,0)--(8,0);
	\draw[arrows=->](0,0)--(0,5);
	\draw[red, thick](0,4)--(3,4);
	\draw[red, thick](6,0)--(6,2);
	\draw[red, thick](3,4)--(6,2);
	%Next MAC figure
	\draw[blue, thick](0,2.5)--(2,2.5);
	\draw[blue, thick](4.5,0)--(4.5,1);
	\draw[blue, thick](2,2.5)--(4.5,1);
	\node[text width=4cm, font=\small] at (2.2,3.3) {$I(X_1;Y\lvert X_2)-I(X_1;Z)$};
	\draw[arrows=->](0.4,3.0)--(0.1,2.7);
	\node[text width=4cm, font=\small] at (2.2,1){$I(X_2;Y\lvert X_1)-I(X_2;Z)$};
	\draw[arrows=->](3.6,1)--(4.4,0.1);
	\node[text width=2cm, font=\small] at (7.2,0.7){$I(X_1;Y\lvert X_2)$};
	\draw[arrows=->] (6.5,0.52)--(6.1,0.1);
	\node[text width=2cm,font=\small] at (1.2,4.5){$I(X_2;Y\lvert X_1)$};
	\draw[arrows=->](0.3,4.4)--(0.1,4.1);
	\node[text width=1cm, font=\small] at (7.5,-0.3){$R_1$};
	\node[text width=1cm, font=\small] at (0.6,5){$R_2$};
\end{tikzpicture}
\caption{Capacity region and Secrecy Rate region of MAC}
\label{MAC_region_compare}
\end{figure}
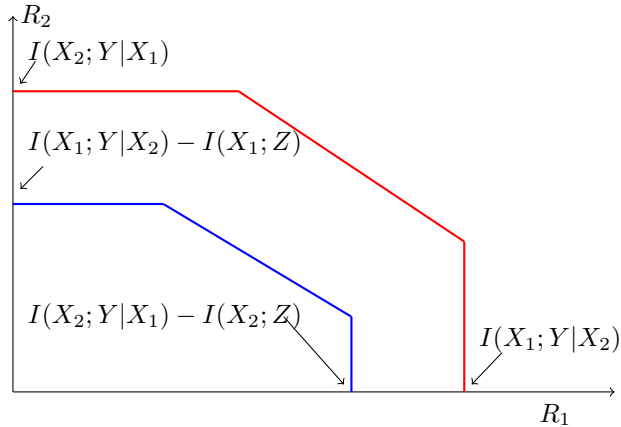

\section{Enhancing the Secrecy-Rate Region of MAC-WT}
\label{section_enhancing_rate}
In this section we extend the coding-decoding scheme of \cite{shah2013previous} for a point-to-point channel to enhance the achievable secrecy rates for a MAC-WT. We recall that in \cite{shah2013previous} the system is slotted with a slot consisting of $n$ channel uses. The first message is transmitted by using the wiretap code of \cite{wyner1975wire} in slot 1. In the next slot we use the message transmitted in slot 1 as a key along with wiretap code and transmit two messages in that slot (keeping the number of channel uses same). Hence the secrecy-rate gets doubled. We continue to use the message transmitted in the previous slot as a key and wiretap coding, increasing the transmission rate till we achieve a secrecy rate equal to the main channel capacity. From then onwards we use only the previous message as key and no wiretap coding. This scheme guarantees that the message which is currently being transmitted is secure w.r.t. all the Eve's outputs, i.e., if message $W_k$ is transmitted in slot $k$ then 
\begin{equation}
\frac{1}{n}I(W_k;\overline{Z}_1,\ldots,\overline{Z}_{k})\rightarrow 0,  \label{leakage_slot_k_1}
\end{equation}
as the codeword length $n \rightarrow \infty$, where $\overline{Z}_i$ is the data received by Eve in slot $i$.

In the following, not only we extend this coding scheme to a MAC-WT but also modify it so that it can be used to improve its secrecy criterion (\ref{leakage_slot_k_1}) and for fading channels as well. The secrecy criterion used  is the following: If user $i$ transmits message $\overline{W}_k^{(i)}$ in slot $k$, we need
\begin{equation}
I(\overline{W}_l^{(1)},\overline{W}_l^{(2)};\overline{Z}_1,\ldots,\overline{Z}_{k}) \leq n \epsilon,~\text{for}~l=1,\ldots,k, \label{leakage_slot_k_induction}
\end{equation}
for any given $\epsilon>0$. This will be strengthened to strong secrecy, $I(\overline{W}_l^{(1)},\overline{W}_l^{(2)};\overline{Z}_1,\ldots,\overline{Z}_{k}) \rightarrow 0~\text{as}~n\rightarrow \infty$ at the end of the section (See the next section for further strengthening of their criteria).  We modify message sets and encoders and decoders with respect to Section \ref{section_channel_model} as follows.

Each slot has $n$ channel uses and is divided into two parts. The first part has $n_1$ channel uses and the second $n_2$, $n_1+n_2=n$. The message sets are $\mathcal{W}^{(i)}=\{1,\ldots,2^{n_1R^s_i}\}$ for users $i=1,2$, where $(R^s_1,R^s_2)$ satisfy (\ref{MAC_WT_Secrecy_region}) for some $(X_1,X_2)$.  The encoders have two parts for both users,
\begin{align}
f^s_1:&\mathcal{W}^{(1)}\rightarrow \mathcal{X}_1^{n_1}, ~~ f^d_1:\mathcal{W}^{(1)}\times \mathcal{K}_1\rightarrow \mathcal{X}_1^{n_2} \\
f^s_2:&\mathcal{W}^{(2)}\rightarrow \mathcal{X}_2^{n_1}, ~~ f^d_2:\mathcal{W}^{(2)}\times \mathcal{K}_2\rightarrow \mathcal{X}_2^{n_2},
\end{align}
where $X_i\in \mathcal{X}_i, i=1,2$, and $\mathcal{K}_i, i=1,2$ are the sets of secret keys generated for the respective user, $f^s_i, i=1,2$ are the wiretap encoders corresponding to each user as in \cite{tekin2008gaussian} and $f_i^d, i=1,2$ are the usual deterministic encoders corresponding to each user in the usual MAC. User $i$ may transmit multiple messages from $\mathcal{W}^{(i)}$ in a slot. In the first part of each slot of $n_1$ length, one message from $\mathcal{W}^{(i)}$ may be transmitted using wiretap coding via $f_i^s$ (denoted by $\overline{W}_{k,1}^{(i)}$ in slot $k$) and in the second part multiple messages from $\mathcal{W}^{(i)}$ may be transmitted (denoted by $\overline{W}_{k,2}^{(i)}$) using messages transmitted in previous slots as keys. The overall message transmitted in slot $k$ by user $i$ is $\overline{W}_k^{(i)}=(\overline{W}_{k,1}^{(i)},\overline{W}_{k,2}^{(i)})$.

The following is our main result.

\begin{thm}
	The secrecy-rate region satisfying (\ref{leakage_slot_k_induction}) is the usual MAC region without Eve, i.e., it is the closure of convex hull of all rate pairs $(R_1,R_2)$ satisfying
\begin{align}
R_1&< I(X_1;Y|X_2), \nonumber \\
R_2&< I(X_2;Y|X_1),  \nonumber \\
R_1+R_2&< I(X_1,X_2;Y), \label{MAC_capacity_region}
\end{align}
for some independent random variables $X_1,X_2$.
\end{thm}

\begin{proof}
	 We fix distributions $p_{X_1},p_{X_2}$. Initially we take $n_1=n_2=n/2$. In slot 1, user $i$ selects message $\mathcal{W}_1^{(i)} \in \mathcal{W}^{(i)}$ to be transmitted confidentially in the first part of the slot, while the second part is not used. Both the users use the wiretap coding scheme of \cite{tekin2008gaussian}. Hence the rate pair $(R_1,R_2)$ satisfies (\ref{MAC_WT_Secrecy_region}) and $R_{L,i}^{(n)}\leq n_1\epsilon,i=1,2$.
In slot 2, the two users select two messages each, $(\overline{W}^{(1)}_{2,1},\overline{W}^{(1)}_{2,2})$ and $(\overline{W}^{(2)}_{2,1},\overline{W}^{(2)}_{2,2})$ to be transmitted. Both users use the wiretap coding scheme (as in \cite{tekin2008gaussian}) for the first part of the message, i.e., $(\overline{W}^{(1)}_{2,1},\overline{W}^{(2)}_{2,1})$, and for the second part user $i$ first takes $XOR$ of $\overline{W}^{(i)}_{2,2}$ with the previous message, i.e., $\overline{W}^{(i)}_{2,2}\oplus \overline{W}^{(i)}_1$. This $XOR$ed message is transmitted over the MAC-WT using a usual MAC coding scheme (\cite{ahlswede1973multi}, \cite{liao1972multiple}). Hence the secure rate achievable in both parts of slot 2 satisfies (\ref{MAC_WT_Secrecy_region}) for both the users. This is also the overall rate of slot 2.

In slot 3, in the first part the rate satisfies (\ref{MAC_WT_Secrecy_region}) via wiretap coding. But in the second part we $XOR$ with $\overline{W}_2^{(i)}$ and are able to send two messages and hence \textit{double} the rate of (\ref{MAC_WT_Secrecy_region}) (assuming $2(R_1,R_2)$ via (\ref{MAC_WT_Secrecy_region}) is within the range of (\ref{MAC_capacity_region})). We continue like this (Fig \ref{leakage_slotwise}).

Define
\begin{equation}
\lambda_1 \triangleq \left\lceil\frac{I(X_1;Y\lvert X_2)}{I(X_1;Y\lvert X_2)-I(X_1;Z)}\right\rceil,  \label{lambda_A_1}
\end{equation}
where $\lceil x\rceil$ is the smallest integer $\geq x$. In slot $\lambda_1+1$ the rate of user 1 in the second part of the slot satisfies,
\begin{align}
R_1 & \leq \min \left( \lambda_1 \left( I(X_1;Y\lvert X_2)-I(X_1;Z) \right), I(X_1;Y\lvert X_2)\right) \nonumber \\
&=I(X_1;Y\lvert X_2).
\end{align}
Similarly we define $\lambda_2$ as
\begin{equation}
\lambda_2 \triangleq \left \lceil\frac{I(X_2;Y|X_1)}{I(X_2;Y|X_1)-I(X_2;Z)} \right \rceil. \label{lambda_A_1}
\end{equation}
In slot $\lambda_2+1$, the rate $R_2$ satisfies
\begin{align}
R_2 \leq I(X_2;Y|X_1).
\end{align}
In slot $\lambda=\max\{\lambda_1,\lambda_2\}+1$, the sum-rate will satisfy
\begin{align}
R&_1+R_2 \leq  \min \left\lbrace \lambda \left[I(X_1,X_2;Y)-\sum_{i=1}^2 I(X_i;Z)\right], I(X_1,X_2;Y)\right\rbrace.
\end{align}
After some slot, say, $\lambda^*>\lambda$, the sum-rate will get saturated by sum-capacity term, i.e., $I(X_1,X_2;Y)$, and hence thereafter the rate pair $(R_1,R_2)\triangleq (R_1^*,R_2^*)$ in the second part of the slot will be at a boundary point of (\ref{MAC_capacity_region}) and the overall rate for the slot is the average in the first part and the second part of the slot.

In slot $k$, (where $k>\lambda^*$) to transmit a message pair $(\overline{W}^{(1)}_k, \overline{W}^{(2)}_k)$, where $\overline{W}^{(i)}_k=(\overline{W}^{(i)}_{k,1},\overline{W}^{(i)}_{k,2}),$ $ i=1,2$, we use wiretap coding for $(\overline{W}^{(1)}_{k,1},\overline{W}^{(2)}_{k,1})$ and for the second part, we $XOR$ it with the previous message i.e., $\overline{W}^{(i)}_{k,2}\oplus \overline{W}^{(i)}_{k-1,2}, i=1,2$, and transmit the overall codeword over the MAC-WT. (Fig. \ref{leakage_slotwise})

 To get the overall rate of a slot close to that in (\ref{MAC_capacity_region}), we make $n_2=ln_1$. By taking $l$ large enough, we can come arbitrarily close to the boundary of (\ref{MAC_capacity_region}).
 
  For this coding scheme, $P_e^n\rightarrow 0$.  A convex combination of the rates in (\ref{MAC_capacity_region}) can be obtained by time sharing.
 Now we show that our coding/decoding scheme also satisfies (\ref{leakage_slot_k_induction}).
\\

%%%**************************FIGURE*************************************************************
%%%%%%%%%%%%%%%%%%%%%%%%%%%%%%%%%%%%%%%%%%%%%%%%%%%%%%%%%%%%%%%%%%%%%%%%%%%%%%%%%%%%%%%%%%%%%%%%%%%

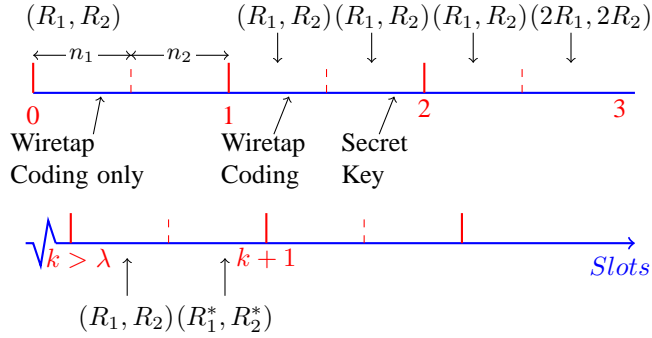
\begin{figure}
		\begin{tikzpicture}
	
	\tikzset{input/.style={}}
	\tikzset{block/.style={rectangle,draw}}
	\tikzstyle{pinstyle} = [pin edge={to-,thick,black}]
	
	\tikzstyle{ann} = [fill=white,font=\footnotesize,inner sep=1pt]
	\draw[blue,thick] (0,0) -- (8,0);
	\draw[blue,thick] (-0.1,-2.0)--(0,-2.0)--(0.1,-2.3)--(0.2,-1.7)--(0.3,-2.0);
	\draw[blue,thick, arrows=->](0.3,-2)--(8,-2);    
	\draw[red,thick] (0,0) -- (0,0.4);
	\draw[red,thin,dashed](1.3,0)--(1.3,0.4);
	\draw[red,thick] (2.6,0) -- (2.6,0.4);
	\draw[red,thin,dashed](3.9,0)--(3.9,0.4);
	\draw[red,thick] (5.2,0) -- (5.2,0.4);
	\draw[red,thin,dashed](6.5,0)--(6.5,0.4);
	%Next line
	\draw[red,thick] (0.5,-2.0) -- (0.5,-1.6);
	\draw[red,thin,dashed] (1.8,-2.0)--(1.8,-1.6);
	\draw[red,thick] (3.1,-2.0) -- (3.1,-1.6);         
	\draw[red,thin,dashed] (4.4,-2.0) -- (4.4,-1.6);         	
	\draw[red,thick](5.7,-2.0)--(5.7,-1.6);
	%\draw[red,thick] (10,0) -- (10,0.4);       
	\node[text width=1.8cm, font=\small ] at (0.6,-0.9) {Wiretap Coding only};     
	\draw[arrows=->] (0.8,-0.55)--(0.9,-0.05);
	\draw[arrows=<->](0,0.5)--(1.3,0.5);   
	\node[ann] at (0.65,0.5) {$n_1$}; 
	\draw[arrows=<->](1.3,0.5)--(2.6,0.5);   
	\node[ann] at (1.95,0.5) {$n_2$}; 
	\node[text width=1.2cm, font=\small] at (3.1,-0.9) {Wiretap Coding}; 
	\node[text width=1.2cm, font=\small] at (4.7,-0.9) {Secret Key};   
	\draw[arrows=->] (3.2,-0.55)--(3.4,-0.05);
	\draw[arrows=->] (4.6,-0.55)--(4.8,-0.05);
	\node[red, font=\small] at (0,-0.3) {0};     \node[red,font=\small] at (2.6,-0.3) {1};  
	\node[red, font=\small] at (5.2,-0.2) {2};     \node[red,font=\small] at (7.8,-0.2) {3}; 
	%Next line 
	\node[red, font=\small] at (0.6,-2.2) {$k>\lambda$};     \node[red,font=\small] at (3.1,-2.2) {$k+1$};    
	\node[blue,font=\small] at (7.8,-2.3) {$Slots$};   
	\node[text width=1.2cm, font=\small] at (0.5,1){$(R_1,R_2)$};        
	\node[text width=1.2cm, font=\small ] at (3.3,1) {$(R_1,R_2)$};      
	\node[text width=1.2cm, font=\small ] at (4.6,1) {$(R_1,R_2)$};        
	\node[text width=1.2cm, font=\small ] at (5.9,1) {$(R_1,R_2)$};    
	\node[text width=1.2cm, font=\small ] at (7.2,1) {$(2R_1,2R_2)$};  
	%Next line
	
	\node[text width=1.2cm, font=\small ] at (1.2,-3.0) {$(R_1,R_2)$};  
	\node[text width=1.2cm, font=\small ] at (2.5,-3.0) {$(R_1^{*},R_2^{*})$};          
	\draw[arrows=<-](3.25,0.4)--(3.25,0.8);   
	\draw[arrows=<-](5.85,0.4)--(5.85,0.8);   
	\draw[arrows=<-](4.5,0.4)--(4.5,0.8);       
	\draw[arrows=<-](7.15,0.4)--(7.15,0.8);
	%Next line    
	\draw[arrows=->](1.25,-2.7)--(1.25,-2.2);          
	\draw[arrows=->](2.55,-2.7)--(2.55,-2.2);       
	\end{tikzpicture}
\caption{Coding Scheme to achieve Shannon Capacity region in MAC}
		\label{leakage_slotwise}
\end{figure}

%%%%%%%%%%%%%%%%%%%%%%%%%%%%%%%%%%%%%%%%%%%%%%%%%%%%%%%%%%%%%%%%%%%%%%%%%%%%
%%%%%%%%%%%%%%%%%%%%%%%%%%%%%%%%%%%%%%%%%%***************************************

\textit{Leakage Rate Analysis}: 
Before we proceed, we define the notation to be used here. For user $i$, the codeword sent in slot $k$ will be represented by $\overline{X}^{(i)}_k$. Correspondingly, $\overline{X}^{(i)}_{k,1}$ and $\overline{X}^{(i)}_{k,2}$ will represent $n_1$-length and $n_2$-length codewords of user $i$ in slot $k$. When we consider $i$ to be 1 or 2, $\bar{i}$ will be taken as 2 or 1 respectively. In slot $k$, the noisy codeword received by Eve is $\overline{Z}_{k}\equiv (\overline{Z}_{k,1},\overline{Z}_{k,2})$, where $\overline{Z}_{k,1}$ is the sequence corresponding to the wiretap coding part and $\overline{Z}_{k,2}$ is corresponding to the $XOR$ part (in which the previous message is used as a key). 

In slot 1, since wiretap coding of \cite{tekin2008gaussian} is employed, the leakage rate satisfies,
\begin{equation}
I(\overline{W}^{(1)}_1;\overline{Z}_1\lvert \overline{X}_1^{(2)})\leq n_1\epsilon,~I(\overline{W}^{(2)}_1;\overline{Z}_1\lvert \overline{X}_1^{(1)})\leq n_1\epsilon.
\end{equation}

For slot 2 we show, for user 1,
\begin{align}
I(\overline{W}^{(1)}_1;\overline{Z}_1,\overline{Z}_2\lvert \overline{X}^{(2)}_2) &\leq n_1\epsilon, \nonumber \\
I(\overline{W}^{(1)}_2;\overline{Z}_1,\overline{Z}_2\lvert \overline{X}^{(1)}_2) &\leq n_1\epsilon. 
\end{align}
Similarly one can show for user 2.
\par
We first note that
\begin{align}
I(&\overline{W}^{(1)}_1;\overline{Z}_1,\overline{Z}_2\lvert \overline{X}^{(2)}_2) \nonumber \\
&=I(\overline{W}^{(1)}_1;\overline{Z}_1)+I(\overline{W}^{(1)}_1;\overline{Z}_2\lvert \overline{Z}_1,\overline{X}^{(2)}_2) \nonumber \\
&\overset{(a)}{\leq} n_1\epsilon+H(\overline{W}^{(1)}_1\lvert \overline{Z}_1,\overline{X}^{(2)}_2)-H(\overline{W}^{(1)}_1\lvert \overline{Z}_1,\overline{X}^{(2)}_2,\overline{Z}_2) \nonumber \\
&\overset{(b)}{=} n_1\epsilon+H(\overline{W}^{(1)}_1\lvert \overline{Z}_1)-H(\overline{W}_1^{(1)}\lvert \overline{Z}_1)=n_1\epsilon.   \label{slot_2_1}
\end{align}
where $(a)$ follows from wiretap coding and $(b)$ follows by the fact that $\overline{X}_2^{(2)} \perp (\overline{W}_1^{(1)},\overline{Z}_1)$, and $(\overline{X}^{(2)}_2,\overline{Z}_2) \perp (\overline{W}_1^{(1)},\overline{Z}_1)$.

Next consider 
\begin{align}
I(&\overline{W}_2^{(1)};\overline{Z}_1,\overline{Z}_2\lvert \overline{X}_2^{(2)})  \nonumber \\
&=I(\overline{W}_{2,1}^{(1)},\overline{W}_{2,2}^{(1)};\overline{Z}_1,\overline{Z}_2\lvert \overline{X}_2^{(2)})  \nonumber \\
&=I(\overline{W}_{2,1}^{(1)};\overline{Z}_1,\overline{Z}_2\lvert \overline{X}_2^{(2)})+I(\overline{W}_{2,2}^{(1)};\overline{Z}_1,\overline{Z}_2\lvert \overline{X}_2^{(2)},\overline{W}_{2,1}^{(1)})  \nonumber \\
&\triangleq I_1+I_2.  \label{leakage_slot_2}
\end{align}
We get upper bounds on $I_1$ and $I_2$. The first term,
\begin{align}
I_1&=I(\overline{W}_{2,1}^{(1)};\overline{Z}_1,\overline{Z}_2\lvert \overline{X}_2^{(2)}) \nonumber \\
&=I(\overline{W}_{2,1}^{(1)};\overline{Z}_1,\overline{Z}_{2,1},\overline{Z}_{2,2}\lvert \overline{X}_2^{(2)}) \nonumber \\
&=I(\overline{W}_{2,1}^{(1)};\overline{Z}_{1}\lvert \overline{X}_2^{(2)})+I(\overline{W}_{2,1}^{(1)};\overline{Z}_{2,1} \lvert \overline{X}_2^{(2)},\overline{Z}_1) \nonumber\\
&+I(\overline{W}_{2,1}^{(1)};\overline{Z}_{2,2} \lvert \overline{X}_2^{(2)},\overline{Z}_1,\overline{Z}_{2,1})  \nonumber \\
&\overset{(a)}{=}0+I(\overline{W}_{2,1}^{(1)};\overline{Z}_{2,1} \lvert \overline{X}_{2,1}^{(2)},\overline{X}_{2,2}^{(2)}, \overline{Z}_1) +I(\overline{W}_{2,1}^{(1)};\overline{Z}_{2,2} \lvert \overline{X}_2^{(2)},\overline{Z}_1,\overline{Z}_{2,1}) \nonumber \\
&\triangleq I_{11}+I_{12},  \label{leakage_slot2_P1}
\end{align}
where $(a)$ follows because $\overline{Z}_1 \perp (\overline{W}_{21}^{(1)},\overline{X}_2^{(2)})$. Furthermore,
\begin{align}
I_{11}&=I(\overline{W}_{2,1}^{(1)};\overline{Z}_{2,1} \lvert \overline{X}_{2,1}^{(2)},\overline{X}_{2,2}^{(2)}, \overline{Z}_1) \nonumber \\
&=H(\overline{W}_{2,1}^{(1)};\lvert \overline{X}_{2,1}^{(2)},\overline{X}_{2,2}^{(2)}, \overline{Z}_1) \nonumber -H(\overline{W}_{2,1}^{(1)};\lvert \overline{Z}_{2,1},\overline{X}_{2,1}^{(2)},\overline{X}_{2,2}^{(2)}, \overline{Z}_1) \nonumber \\
&\overset{(a)}{=}H(\overline{W}_{2,1}^{(1)};\lvert \overline{X}_{2,1}^{(2)})-H(\overline{W}_{2,1}^{(1)};\lvert \overline{Z}_{2,1},\overline{X}_{2,1}^{(2)}) \nonumber \\
&=I(\overline{W}_{2,1}^{(1)};\overline{Z}_{2,1},\lvert \overline{X}_{2,1}^{(2)}) \overset{(b)}{\leq}n_1\epsilon,  \label{leakage_slot_P2}
\end{align}
where $(a)$ follows since $(\overline{X}_{2,2}^{(2)}, \overline{Z}_1) \perp (\overline{W}_{2,1}^{(1)}, \overline{Z}_{2,1},\overline{X}_{2,1}^{(2)})$ and $(b)$ follows since the first part of the message is encoded via the usual coding scheme for MAC-WT. 
\\
$~~$Also,
\begin{align}
I_{12}&=I(\overline{W}_{2,1}^{(1)};\overline{Z}_{2,2} \lvert \overline{X}_2^{(2)},\overline{Z}_1,\overline{Z}_{2,1}) \nonumber \\
&=H(\overline{W}_{2,1}^{(1)};\lvert \overline{X}_{2,1}^{(2)},\overline{X}_{2,2}^{(2)},\overline{Z}_1,\overline{Z}_{2,1}) \nonumber \\ &-H(\overline{W}_{2,1}^{(1)}\lvert \overline{X}_{2,1}^{(2)},\overline{X}_{2,2}^{(2)},\overline{Z}_1,\overline{Z}_{2,1},\overline{Z}_{2,2} ) \nonumber \\
&\overset{(a)}{=}H(\overline{W}_{2,1}^{(1)};\lvert \overline{X}_{2,1}^{(2)}, \overline{Z}_{2,1})-H(\overline{W}_{2,1}^{(1)};\lvert \overline{X}_{2,1}^{(2)},\overline{Z}_{2,1}) =0,  \nonumber
\end{align}
where $(a)$ follows since $(\overline{X}_{2,2}^{(2)},\overline{Z}_1,\overline{Z}_{2,2}) \perp (\overline{W}_{2,1}^{(1)},\overline{X}_{2,1}^{(2)},\overline{Z}_{2,1})$.

From (\ref{leakage_slot_2}), (\ref{leakage_slot2_P1}) and (\ref{leakage_slot_P2}) we have $I_1=I_{11}+I_{12}\leq n_1\epsilon$. 

Next consider,
\begin{align}
I_2&=I(\overline{W}_{2,2}^{(1)};\overline{Z}_1,\overline{Z}_2\lvert \overline{X}_2^{(2)},\overline{W}_{2,1}^{(1)})  \nonumber \\
&=I(\overline{W}_{2,2}^{(1)};\overline{Z}_2\lvert \overline{X}_2^{(2)},\overline{W}_{2,1}^{(1)})  +I(\overline{W}_{2,2}^{(1)};\overline{Z}_1\lvert \overline{X}_2^{(2)},\overline{W}_{2,1}^{(1)},\overline{Z}_2).  \label{leakage_slot2_I_2}
\end{align}
We have,
\begin{align}
I&(\overline{W}_{2,2}^{(1)};\overline{Z}_2\lvert \overline{X}_2^{(2)},\overline{W}_{2,1}^{(1)}) \nonumber \\
&=I(\overline{W}_{2,2}^{(1)};\overline{Z}_{2,1}\lvert \overline{X}_2^{(2)},\overline{W}_{2,1}^{(1)})  +I(\overline{W}_{2,2}^{(1)};\overline{Z}_{2,2}\lvert \overline{X}_2^{(2)},\overline{W}_{2,1}^{(1)},\overline{Z}_{2,1})  \nonumber \\
&\overset{(a_1)}{=}0+I(\overline{W}_{2,2}^{(1)};\overline{Z}_{2,2}\lvert \overline{X}_2^{(2)},\overline{W}_{2,1}^{(1)},\overline{Z}_{2,1}) \nonumber \\
&\overset{(a_2)}{=}I(\overline{W}_{2,2}^{(1)};\overline{Z}_{2,2}\lvert \overline{X}_{2,2}^{(2)}) \overset{(a_3)}{=}0,  \nonumber
\end{align}
 and $(a_1)$ follows since $\overline{W}_{2,2}^{(1)} \perp (\overline{Z}_{2,1},\overline{X}_2^{(2)},\overline{W}_{2,1}^{(1)})$; $(a_2)$ holds because $(\overline{X}_{2,1}^{(2)},\overline{W}_{2,1}^{(1)})\perp (\overline{W}_{2,2}^{(1)},\overline{Z}_{2,2},\overline{X}_{2,2}^{(2)})$; and $(a_3)$ is true since $\overline{W}_{2,2}^{(1)} \perp (\overline{X}_{2,2}^{(2)},\overline{Z}_{2,2})$.

 Also,
 \begin{align}
I&(\overline{W}_{2,2}^{(1)};\overline{Z}_1\lvert \overline{X}_2^{(2)},\overline{W}_{2,1}^{(1)},\overline{Z}_2)  \nonumber \\
&=I(\overline{W}_{2,2}^{(1)};\overline{Z}_1\lvert \overline{X}_{2,1}^{(2)},\overline{X}_{2,2}^{(2)},\overline{W}_{2,1}^{(1)},\overline{Z}_{2,1},\overline{Z}_{2,2}) \nonumber \\
&\overset{(b_1)}{=}I(\overline{W}_{2,2}^{(1)};\overline{Z}_1\lvert \overline{X}_{2,2}^{(2)},\overline{Z}_{2,2}) \overset{(b_2)}{=}0,  \nonumber
 \end{align}
 where $(b_1)$ follows since $(\overline{W}_{2,1}^{(1)},\overline{Z}_{2,1},\overline{X}_{2,1}^{(2)}) \perp (\overline{Z}_{2,2},\overline{X}_{2,2}^{(2)},\overline{W}_{2,2}^{(1)},\overline{Z}_1)$), and $(b_2)$ follows because $\overline{Z}_1\perp (\overline{W}_{2,2}^{(1)},\overline{X}_{2,2}^{(2)},\overline{Z}_{2,2})$.
 Hence from (\ref{leakage_slot2_I_2}) we have $I_2=0$.
 
 From (\ref{leakage_slot_2}) we have
\begin{equation}
I(\overline{W}_2^{(1)};\overline{Z}_1,\overline{Z}_2\lvert \overline{X}_2^{(2)}) \leq n_1\epsilon.
\end{equation}
Similarly one can show that 
\begin{equation}
I(\overline{W}_2^{(2)};\overline{Z}_1,\overline{Z}_2\lvert \overline{X}_2^{(1)}) \leq n_1\epsilon.
\end{equation}
Therefore, from (\ref{leakage_2}),
\begin{align}
&I(\overline{W}^{(1)}_2,\overline{W}^{(2)}_2;\overline{Z}_1,\overline{Z}_2)\nonumber \\
&\leq  I(\overline{W}^{(1)}_2;\overline{Z}_1,\overline{Z}_{2}\lvert \overline{X}^{(2)}_2)+I(\overline{W}^{(2)}_2;\overline{Z}_1,\overline{Z}_{2}\lvert \overline{X}^{(1)}_2). \nonumber
\end{align}

To prove that (\ref{leakage_slot_k_induction}) holds for any slot, we use mathematical induction in the following lemma. For a proof, please see \cite{shah2014achieving_MAC}.

 \begin{lem}
 	Let (\ref{leakage_slot_k_induction}) hold for $k$, then it also holds for $k+1$.
 \end{lem}

 \end{proof}

\begin{rem}[\textbf{A note about strong secrecy}]
	The notion of secrecy used above is \textit{weak secrecy}, i.e., if message $W$ is transmitted and Eve receives $Z^n$, then $I(W;Z^n)\leq n_1\epsilon $. \textit{Strong Secrecy} requires that $I(W;Z^n)\leq \epsilon$.
In single user case, if strong secrecy notion is used instead of weak secrecy, the secrecy capacity does not change (\cite{maurer2000information}). The same result has been proved for a multiple access channel with a wiretapper in \cite{yassaee2010multiple} using the channel resolvability technique. In our coding scheme of Theorem 2 if we use resolvability based coding in slot 1, and in subsequent slots use both resolvability based coding (in the first part of the slot) and the previous message (which is now strongly secure w.r.t. Eve) as a key in the second part of the slot, we can achieve the same secrecy-rate region (capacity region of usual MAC without Eve), satisfying the leakage rate
\begin{align}
\limsup_{n\rightarrow \infty}I(\overline{W}_k^{(1)},\overline{W}_k^{(2)};\overline{Z}_1,\overline{Z}_2,\ldots,\overline{Z}_{k}) = 0,
\end{align}
as $n\rightarrow \infty$, because in the RHS of (\ref{leakage_slot_k_induction}), we can get $\epsilon$ instead of $2n_1\epsilon$.
\end{rem}

\section{Discrete Memoryless MAC-WT with Buffer}
\label{section_mac_with_buffer}
\thinmuskip=0mu
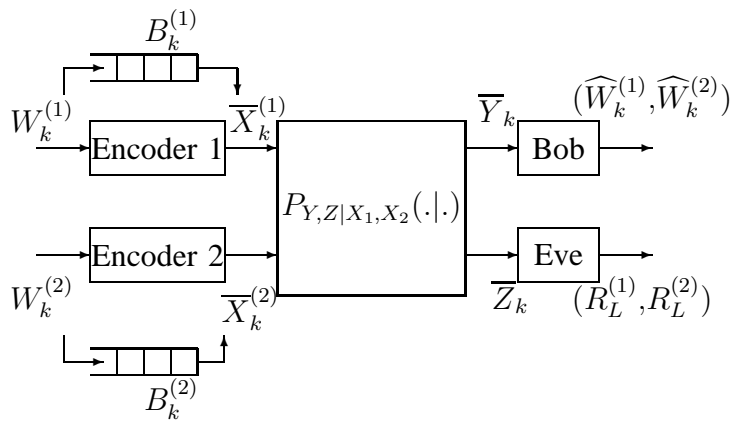
\begin{figure}[h]
	\setlength{\unitlength}{0.14in} % selecting unit length
	\centering % used for centering Figure
	\begin{picture}(32,15) % picture environment with the size (dimensions)
	% 32 length units wide, and 15 units high.
	\put(2,8.5){\framebox(5,2){Encoder 1}}
	\put(2,4.5){\framebox(5,2){Encoder 2}}
	\put(9,4){\framebox(7,6.5){$P_{Y,Z\lvert X_1,X_2}(.\lvert .)$}}
	\put(18,8.5){\framebox(3,2){Bob}}
	\put(18,4.5){\framebox(3,2){Eve}}
	\put(0,9.5){\vector(1,0){2}}\put(7,9.5){\vector(1,0){2}}
	\put(7,5.5){\vector(1,0){2}}
	\put(0,5.5){\vector(1,0){2}}
	\put(16,9.5){\vector(1,0){2}}\put(16,5.5){\vector(1,0){2}}
	\put(21,9.5){\vector(1,0){2}}
	\put(21,5.5){\vector(1,0){2}}
	\put(-1,10.1) {$W^{(1)}_k$}
	\put(-1,3.5) {$W^{(2)}_k$}
	\put(20,11) {$(\widehat{W}^{(1)}_k,\widehat{W}^{(2)}_k)$}
	\put(20,3.5) {$(R_L^{(1)},R_L^{(2)})$}
	\put(7.2,10.1)
	{$\overline{X}^{(1)}_k$} \put(6.9,3.0){$\overline{X}^{(2)}_k$}\put(16.5,10.5) {$\overline{Y}_k$}
	\put(17,3.5){$\overline{Z}_k$}
	%Buffer
	\put(2,1){\line(1,0){4}}
	\put(2,2){\line(1,0){4}}
	\put(3,1){\line(0,1){1}}
	\put(4,1){\line(0,1){1}}
	\put(5,1){\line(0,1){1}}
	\put(6,1){\line(0,1){1}}	
	\put(1,1.5){\vector(1,0){1.5}}		
% second buffer
	\put(2,12){\line(1,0){4}}
	\put(2,13){\line(1,0){4}}
	\put(3,12){\line(0,1){1}}
	\put(4,12){\line(0,1){1}}
	\put(5,12){\line(0,1){1}}
	\put(6,12){\line(0,1){1}}	
	\put(1,12.5){\vector(1,0){1.5}}
	
	%for arrows
	%\put(1,12.5){\vector(1,0){1.5}}
		\put(1,11.5){\line(0,1){1}}
		\put(6,12.5){\line(1,0){1.5}}
		\put(7.5,12.5){\vector(0,-1){1}}
		\put(4,13.5){$B_k^{(1)}$}
		\put(4,-0.3){$B_k^{(2)}$}
		%\put(8,11){$\overline{R}_k$}
	%second buffew arror
	\put(1,1.5){\line(0,1){1}}
	\put(6,1.5){\line(1,0){1}}
	\put(7,1.5){\vector(0,1){1}}
	\end{picture}
	\caption{Discrete Memoryless Multiple Access Wiretap Channel with secret key buffers} % title of the Figure
	\label{fig:DM_MAC_BUFFER} % label to refer figure in text
\end{figure}

In this section we improve the result in Theorem 3.1 by obtaining rates (\ref{MAC_capacity_region}) while enhancing the secrecy requirement from  (\ref{leakage_slot_k_induction}) to 
\begin{align}
I(\overline{W}^{(1)}_{k}, \overline{W}^{(1)}_{k-1}, \ldots, \overline{W}^{(1)}_{k-N_1};\overline{Z}_1,\ldots,\overline{Z}_k \lvert \overline{X}_k^{(2)}) &\leq n_1\epsilon, \nonumber \\
I(\overline{W}^{(2)}_{k}, \overline{W}^{(2)}_{k-1}, \ldots, \overline{W}^{(1)}_{k-N_1};\overline{Z}_1,\ldots,\overline{Z}_k \lvert \overline{X}_k^{(1)}) &\leq n_1\epsilon, \nonumber \\
I(\overline{W}^{(1)}_{k}, \overline{W}^{(2)}_{k}, \ldots, \overline{W}^{(1)}_{k-N_1},\overline{W}^{(2)}_{k-N_1};\overline{Z}_1,\ldots,\overline{Z}_k) &\leq 2n_1\epsilon   , \label{leakage_with_buffer}
\end{align} 
where $N_1$ can be arbitrarily large. This will satisfy the requirements of any practical system \footnote[1]{In many countries, confidential messages beyond a certain period are declassified by law.}. For this, we use a key buffer at each of the users and instead of using the messages transmitted in slot $(k-1)$ as the key in slot $k$, we use the messages transmitted in slots before $k-N_1-1$.

Let each user has an infinite key buffer to store the key bits. The message $\overline{W}_k^{(i)}$ after transmission in slot $k$ from user $i$ is stored in its key buffer at the end of the slot. However now in slot $k+1$ we use the \textit{oldest} bits stored in its key buffer as a key in the second part of its slot. Once certain bits from the key buffer have been used as a key, these are discarded from the key buffer. Let $B_k^{(i)}$ be the number of key bits in the key buffer of the $i^{th}$ user at the beginning of the $k^{th}$ slot. Then out of this, for $k\geq \lambda^{*}$, the number of key bits used in a slot by user 1 is $C_1 n_2$ (since these are used only in the second part of the slot) where $C_1\leq I(X_1;Y\lvert X_2)$, while the total number of secret bits transmitted in the slot is $C_1n_2+R_s^{(1)}n_1$. These transmitted bits also get stored in its key buffer at time $k+1$. Similarly it holds for user 2. Thus $B_k^{(i)} \rightarrow \infty$ as $k\rightarrow \infty$ for $i=1,2$.

After some time (say $N_2$ slots) since we are using the oldest bits in the key buffer, for $k\geq N_2$, we will be using the secret key bits only from messages $(\overline{W}_1^{(i)},\overline{W}_2^{(i)},\ldots,\overline{W}_{k-N_1-1}^{(i)})$ for securing messages $(\overline{W}_k^{(i)},\overline{W}_{k-1}^{(i)},\ldots,\overline{W}_{k-N_1}^{(i)})$, for user $i=1,2$ respectively. The following proof works for $N_1>0$. Theorem 2.1 wa for $N_1=0$.

\begin{thm}
	 The secrecy-rate region (with the leakage rate constraints (\ref{leakage_with_buffer})) of a DM-MAC-WT equals the usual Shannon capacity region (\ref{MAC_capacity_region}) of the MAC.
\end{thm}

\begin{proof}
	 With the proposed modification of this section to the coding-decoding scheme of Section 3, in any slot $k$, the legitimate receiver is able to decode the message pair $(\overline{W}_k^{(1)},\overline{W}_k^{(2)})$ with probability of error $P_e^{(n)}\rightarrow 0$ as $n\rightarrow \infty$. Also (\ref{leakage_slot_k_induction}) along with $R_{L,i}^{(n)}\leq n_1\epsilon_1, i=1,2$ continue to be satisfied, where $\epsilon_1>0$ will be fixed later on.

Now we consider the leakage rate. We have,
\begin{align}
I&(\overline{W}^{(1)}_{k}, \overline{W}^{(1)}_{k-1}, \ldots, \overline{W}^{(1)}_{k-N_1};\overline{Z}_1,\ldots,\overline{Z}_k \lvert \overline{X}_k^{(2)}) \nonumber \\
&=I(\overline{W}^{(1)}_{k,1}, \overline{W}^{(1)}_{k-1,1}, \ldots, \overline{W}^{(1)}_{k-N_1,1};\overline{Z}_1,\ldots,\overline{Z}_k \lvert \overline{X}_k^{(2)}) \nonumber \\
&+I(\overline{W}^{(1)}_{k,2}, \overline{W}^{(1)}_{k-1,2}, \ldots, \overline{W}^{(1)}_{k-N_1,2};\overline{Z}_1,\ldots,\overline{Z}_k \lvert \overline{X}_k^{(2)},\overline{W}^{(1)}_{k,1}, \ldots, \overline{W}^{(1)}_{k-N_1,1}). \nonumber \\
\end{align}
From Lemma \ref{lemma_a1} and Lemma \ref{lemma_a2} in the Appendix,

\begin{align}
I(\overline{W}^{(1)}_{k,1}, \overline{W}^{(1)}_{k-1,1}, \ldots, \overline{W}^{(1)}_{k-N_1,1};\overline{Z}_1,\ldots,\overline{Z}_k \lvert \overline{X}_k^{(2)}) \leq n_1\epsilon \label{lemma_1}
\end{align}
and
\begin{align}
I&(\overline{W}^{(1)}_{k,2}, \overline{W}^{(1)}_{k-1,2}, \ldots, \overline{W}^{(1)}_{k-N_1,2};\overline{Z}_1,\ldots,\overline{Z}_k \lvert \overline{X}_k^{(2)},  \overline{W}^{(1)}_{k,1}, \ldots, \overline{W}^{(1)}_{k-N_1,1})\leq 6n_1\epsilon \label{lemma_2}
\end{align}
Thus, taking $\epsilon = \epsilon/7,$ we obtain the first inequality in (\ref{leakage_with_buffer}). Similarly we can show the second inequality.

To prove the third inequality we define $\widetilde{W}^{(1)}\triangleq (\overline{W}^{(1)}_{k}, \overline{W}^{(1)}_{k-1}, \ldots, \overline{W}^{(1)}_{k-N_1})$, $\widetilde{W}^{(2)}\triangleq(\overline{W}^{(2)}_{k}, \overline{W}^{(2)}_{k-1}, \ldots, \overline{W}^{(2)}_{k-N_1})$ and $\widetilde{Z}\triangleq (\overline{Z}_1,\ldots,\overline{Z}_k)$, and we have
\begin{align}
I(&\widetilde{W}^{(1)},\widetilde{W}^{(2)};\widetilde{Z})  \nonumber \\
&=I(\widetilde{W}^{(1)};\widetilde{Z})+I(\widetilde{W}^{(2)};\widetilde{Z}\lvert \widetilde{W}^{(1)}) \nonumber \\
&=H(\widetilde{W}^{(1)})-H(\widetilde{W}^{(1)}\lvert \widetilde{Z})+H(\widetilde{W}^{(2)})  -H(\widetilde{W}^{(2)}\lvert \widetilde{Z},\widetilde{W}^{(1)}) \nonumber \\
&\overset{(a)}{\leq}H(\widetilde{W}^{(1)}\lvert \overline{X}_k^{(2)})-H(\widetilde{W}^{(1)}\lvert \widetilde{Z} ,\overline{X}_k^{(2)})+H(\widetilde{W}^{(2)}\lvert \overline{X}_k^{(1)})  -H(\widetilde{W}^{(2)}\lvert \widetilde{Z},\overline{X}_k^{(1)}) \nonumber \\
&=I(\widetilde{W}^{(1)};\widetilde{Z}\lvert \overline{X}_k^{(2)})+I(\widetilde{W}^{(2)};\widetilde{Z}\lvert \overline{X}_k^{(1)}),  \label{leakage_2_buffer}
\end{align}
where $(a)$ follows from the facts: conditioning decreases entropy, messages are independent and a codeword is a function of the message to be transmitted.
Hence, from (\ref{lemma_1}) and (\ref{lemma_2})
\begin{align}
I(&\overline{W}^{(1)}_{k}, \overline{W}^{(2)}_{k}, \ldots, \overline{W}^{(1)}_{k-N_1},\overline{W}^{(2)}_{k-N_1};\overline{Z}_1,\ldots,Z_k^n) \leq n_1\epsilon. 
\end{align}
\end{proof}

\section{Fading MAC-WT}
\label{section_fading_mac}

In this section we consider a two user discrete time additive white Gaussian fading channel. If $X_1,~X_2$ are the channel inputs, then Bob receives
\begin{equation}
Y=\widetilde{H}_1X_2+\widetilde{H}_2X_2+N_1
\end{equation}
and Eve receives
\begin{equation}
Z=\widetilde{G}_1X_1+\widetilde{G}_2X_2+N_2,
\end{equation}
where $\widetilde{H}_i$ is the channel gain to Bob, $\widetilde{G}_i$ is the channel gain to Eve and $N_{i}$ has Gaussian distribution with mean 0 and variance $\sigma_i^2$,  $i=1,2$. We assume that the random variables $\widetilde{H}_1,~ \widetilde{H}_2,~ \widetilde{G}_1, \widetilde{G}_2, N_{1}, N_{2}$ are independent of each other. The channel is experiencing slow fading, i.e., the channel gains remain same during the transmission of the whole codeword. Let $H_i=\lvert\widetilde{H}_i\rvert^2$ and $G_i=\lvert\widetilde{G}_i\rvert^2$, $i=1,2.$ Average power constraint for user $i$ is $\overline{P}_i$.

We define some notation for convenience. For $H=(H_1,H_2),~G=(G_1,G_2)$,
\begin{align}
C_1(P_1(H,G)) & \triangleq \frac{1}{2}\log\left(1+\frac{H_1P_1(H,G)}{\sigma_1^2}\right) \nonumber \\
C_2(P_2(H,G)) & \triangleq \frac{1}{2}\log\left(1+\frac{H_2P_1(H,G)}{\sigma_1^2}\right) \nonumber \\
C_1^e(P_1(H,G)) & \triangleq \frac{1}{2}\log\left(1+\frac{G_1P_1(H,G)}{\sigma_2^2+G_2P_2(H,G)}\right)  \nonumber \\
C_2^e(P_2(H,G)) & \triangleq \frac{1}{2}\log\left(1+\frac{G_2P_2(H,G)}{\sigma_2^2+G_1P_1(H,G)}\right)  \nonumber \\
C(P_1(H,G),P_2(H,G))&\triangleq \frac{1}{2}\log\left(1+\frac{H_1P_1(H,G)+H_2P_2(H,G)}{\sigma_1^2}\right)
\end{align}

An achievable secrecy rate region for this channel is 
\begin{equation}
\mathcal{R}_g^s(\overline{P})=
\end{equation}

\begin{equation}
\left\lbrace
\begin{array}{c l}
(R_1,R_2): \qquad\qquad\qquad\qquad\\
R_1\leq \mathbb{E}_{H,G}\left[ \left(C_1(P_1)-C_1^e(P_1)\right)^+\right] \\
R_2\leq \mathbb{E}_{H,G}\left[ \left(C_2(P_2)-C_2^e(P_2)\right)^+\right] \\
R_1+R_2\leq \mathbb{E}_{H,G}\left[ \left(C(P_1,P_2)-\sum_{i=1}^{2}C_i^e(P_i)\right)^+\right] \\
\end{array}\right\rbrace   \label{GMAC_WT_region}
\end{equation}
where $\overline{P}=(\overline{P}_1,\overline{P}_2)$. To achieve these rates (with $P_i(H,G)\equiv \overline{P}_i$), the transmitters need not know the channel states, but Bob's receiver needs to know all $H_i$, $G_i$. We assume this in this section.

If the channel states $(H,G)$ are known at each of the users as well as at the receiver of Bob, then we can improve over the rate region in (\ref{GMAC_WT_region}) by making the transmit powers as functions of $(H,G)$:
\begin{equation}
\mathcal{P}:H\times G\rightarrow \mathbb{R}^2_{+},
\end{equation}
where $\mathcal{P}=(P_1,P_2)$. Now we denote the rate region as $\mathcal{C}_f^s(\mathcal{P})$. 
We note that the secrecy capacity region of MAC-WT ($\mathcal{C}_f^s(\mathcal{P})$) is not known, but  $\mathcal{R}^s_f(\mathcal{P})\subseteq \mathcal{C}^s_f(\mathcal{P})$ \cite{tekin2007secrecy}.

Now we use the coding scheme of Section 3 to the two user fading MAC-WT to enlarge the secrecy rate region to the usual capacity region of the fading channel.
%%%%%%%%%************************%%%%%%%%%%%%%%%%%%%%%****************
Message pair $(\overline{W}_k^{(1)}, \overline{W}_k^{(2)})$ is to be transmitted confidentially by the two users over the fading MAC in slot $k$, and will be stored in their respective secret key buffers at the end of the $k^{th}$ slot. Let $B_k^{(1)}, B_k^{(2)}$ be the number of bits in the key buffers of users 1 and 2 respectively at the beginning of the slot $k$. Let $\overline{R}_k^{(i)}$ bits be taken from the key buffer of user $i$ to act as a secret key for transmission of message $\overline{W}_k^{(i)}$. The two users satisfy the long term average power constraint
\begin{equation}
\limsup_{k\rightarrow \infty}\frac{1}{k}\sum_{m=1}^{k}\mathsf{E}\left[P_i(H_k,G_k)\right] \leq \overline{P}_i,~~ i=1,2, \label{power_const}
\end{equation}
where $H_k,~G_k$ are the channel gains in slot $k$ and $P_i(H_k,G_k)$ is the average power used by user $i$ in slot $k$. We need to compute $P_i(H,G)$ and $\overline{R}_k^{(i)}, i=1,2$ such that the resulting average rate region $(\overline{r}^{(1)},\overline{r}^{(2)})$ is maximized, where
\begin{equation}
\overline{r}^{(i)}=\limsup_{k\rightarrow \infty}\frac{1}{k}\sum_{l=1}^{k}r_l^{(i)},
\end{equation}
$r_k^{(i)}$ is the transmission rate of user $i$ in slot $k$, subject to the long term respective power constraints (\ref{power_const}). The secrecy-rate region is computed when 
\begin{equation}
Pr\left(\{H_{1k}>G_{1k}\} \cup \{H_{2k}>G_{2k}\} \right)>0,
\end{equation} 
where $Pr(A)$ represents the probability of event $A$.
Otherwise, the secrecy rate region is zero. Actually we state the following theorem for $Pr(H_{ik}>G_{ik})>0,i=1,2.$ If it is not true for any one $i$ then the secrecy rate for that user is zero. For both the transmitting users, at the end of slot $k$, $\widehat{r}_k^{(i)}=n(l+1)r_k^{(i)}$ bits are stored in the secret key buffer for future use as a key, where $n_2=ln_1$. Hence $B_k^{(i)}$ evolves as
\begin{eqnarray}
B_{k+1}^{(i)}=B_k^{(i)}+\widehat{r}_k^{(i)}-\overline{R}_k^{(i)}. \label{buffer_evolution}
\end{eqnarray}
where $\hat{r}_k^{(i)} \geq \overline{R}_k^{(i)}$ and $\hat{r}_k^{(i)}>\overline{R}_k^{(i)}$ with positive probability $Pr(H_{ik}>G_{ik})$. Therefore, $B_k^{(i)} \rightarrow \infty~a.s.$ for $i=1,2$.
 \begin{thm}
 	If $Pr(H_{ik}>G_{ik})>0,~i=1,2,$ and all the channel gains are available at all the transmitters, then the following long term average rates that maintain the leakage rates (\ref{leakage_with_buffer}), are achievable:
\begin{align}
R_1 &\leq \frac{1}{2}\mathsf{E}_{H,G}\left[C_1\left(P_1(H)\right)\right],  \nonumber \\
R_2 &\leq \frac{1}{2}\mathsf{E}_{H,G}\left[C_2\left(P_2(H)\right)\right],  \nonumber \\
R_1+R_2 &\leq \frac{1}{2}\mathsf{E}_{H,G}\left[C\left(P_1(H),P_2(H)\right)\right].  \label{FMAC_CAPACIty}
\end{align}
where $P$ is any policy that satisfies average power constraint. If only Bob knows all the channel states but not the transmitters, then $(R_1,R_2)$ satisfies (\ref{FMAC_CAPACIty}) with $P_i(H,G)\equiv \overline{P}_i,~i=1,2$.
 \end{thm}

%\begin{figure}[!htb]
%	\centering
%	\includegraphics[scale=0.4]{FADING_MAC.pdf}
%	\caption{Achievable rates for MAC-WT with and without secrecy constraints}
%	\label{fig:digraph}
%\end{figure}

\textit{Sketch of Achievability Scheme:} We use the coding-decoding scheme proposed in Section 3 with appropriate changes to account for the fading process. Assuming $B^{(i)}_0=0, ~i=1,2$, user $i$ transmits the first time when $H_{ik}>G_{ik}$. Then it uses the usual MAC wiretap coding as proposed in \cite{tekin2008gaussian} in all its $l+1$ mini-slots.

In the next slot (say $k^{th})$ user $i$ uses the first mini-slot for wiretap coding (if $H_{ik}>G_{ik}$ for user $i$) and the rest of the $m$ mini-slots for transmission via the secret key (if $H_{ik}<G_{ik}$ the first mini-slot is not used). It uses $\overline{R}_k^{(i)}=\min \left (B_k^{(i)}, lC_i(P_i(H,G)n_1)\right)$ key bits which are removed from the key buffer at the end of the slot. The total number of bits transmitted by user $i$ in slot $k$ is
\begin{equation}
\widehat{r}_k^{(i)}=\overline{R}^{(i)}_k+n_1\left(C_i(P_1(H_k,G_k))-C_i^e(P_i(H_k,G_k))\right) ^+.
\end{equation}
These bits are stored in the key buffer at the end of the slot. Thus $\widehat{r}_k^{(i)}\geq \overline{R}_k^{(i)}$ and since $Pr(H_{ik}>G_{ik})>0,~i=1,2$, $Pr(\widehat{r}_k^{(i)}>\overline{R}_k^{(i)})>0$. Thus $B_k^{(i)}\rightarrow \infty~a.s.$ for $i=1,2.$ 

Also, as before, we can show that after some slot $k\geq N_2$, with an arbitrarily large probability, the messages transmitted in slots $k,~ k-1, \ldots, ~k-N_1$ will use the messages transmitted before $k-N_1-1$, and the rate used in the first minislot will satisfy (\ref{GMAC_WT_region}) but the rate used in the second minislot will satisfy (\ref{FMAC_CAPACIty}). The overall rate of the slot can be made as close to (\ref{FMAC_CAPACIty}) as we wish by taking $l$ large. Thus the rest of the proof to show $P_e^n\rightarrow 0$ and that ($\ref{leakage_with_buffer}$) is satisfied follows from Theorem 3.1.

All the above results extend in \textit{strong} secrecy sense as in Section 3, by using the \textit{resolvability} based coding scheme of \cite{yassaee2010multiple} instead of usual wiretap coding for MAC-WT of \cite{tekin2008gaussian}. 
\section{Conclusions}
\label{section_conclusion}

In this paper we obtain the secrecy-rate region for  a slotted multiple access wiretap channel. We show that by using the previous message as a key in the next slot we can achieve secrecy-rate region equal to the capacity region of a MAC, if we consider the secrecy rate of individual messages.
We then extend the result to the case where an arbitrarily large number of recent multiple messages are secure w.r.t. the information of Eve, by using the secret key buffer for both the transmitters. Finally, we further extend our coding scheme to a fading Gaussian channel and show that the usual Shannon capacity region can be obtained while retaining the secrecy of the multiple messages.

\appendix

\section{DM-MAC-WT with secret key Buffer}
\begin{lem}
The following inequality is satisfied 
\begin{equation}
I(\overline{W}^{(1)}_{k,1},~\overline{W}^{(1)}_{k-1,1},~ \ldots,~ \overline{W}^{(1)}_{k-N_1,1};\overline{Z}_1,~\ldots,~\overline{Z}_k \lvert \overline{X}_k^{(2)}) \leq n_1\epsilon.
\end{equation}
\label{lemma_a1}
\end{lem}
\begin{proof}
	We have
\begin{align}
I&(\overline{W}_{k,1}^{(1)},~\overline{W}_{k-1,1}^{(1)},~\ldots,~ \overline{W}_{k-N_1,1}^{(1)}; \overline{Z}_1,~\ldots,~\overline{Z}_k \lvert \overline{X}_k^{(2)})  \nonumber \\
&=I(\overline{W}_{k,1}^{(1)}; \overline{Z}_1,~\ldots,~\overline{Z}_k \lvert \overline{X}_k^{(2)})  \nonumber \\
&+I(\overline{W}_{k-1,1}^{(1)}; \overline{Z}_1,~\ldots,~\overline{Z}_k \lvert \overline{X}_k^{(2)},~\overline{W}_{k,1}^{(1)})  \nonumber \\
&+\ldots+I(\overline{W}_{k-N_1,1}^{(1)}; \overline{Z}_1,~\ldots,~\overline{Z}_k \lvert \overline{X}_k^{(2)},~  \overline{W}_{k,1}^{(1)},~\overline{W}_{k-1,1}^{(1)},~\ldots,~ \overline{W}_{k-N_1+1,1}^{(1)})  \nonumber \\
&\triangleq I_1+I_2+\ldots+I_{N_1}    \label{leakage_part_2}
\end{align}
Now let us evaluate each term. Denoting the two parts of $\overline{Z}_k$ by $\overline{Z}_{k,1},\overline{Z}_{k,2}$, and choosing the wiretap coding with leakage rate $\leq n_1\epsilon_1$, where $\epsilon_1 = \epsilon/N_1$,
\begin{align}
I_1&=I(\overline{W}_{k,1}^{(1)}; \overline{Z}_{1,1},~\overline{Z}_{1,2},~\ldots,~\overline{Z}_{k,1},~\overline{Z}_{k,2} \lvert \overline{X}_k^{(2)})  \nonumber  \\
&=I(\overline{W}_{k,1}^{(1)}; \overline{Z}_{k,1} \lvert \overline{X}_k^{(2)}) +I(\overline{W}_{k,1}^{(1)}; \overline{Z}_1,~\ldots,~\overline{Z}_{k-1},~\overline{Z}_{k,2} \lvert \overline{X}_k^{(2)})  \nonumber \\
&\overset{(a)}{\leq}n_1\epsilon_1 + I(\overline{W}_{k,1}^{(1)}; \overline{Z}_1,~\ldots,~\overline{Z}_{k-1},~\overline{Z}_{k,2} \lvert \overline{X}_k^{(2)}, )  \nonumber \\
&=n_1\epsilon_1+H(\overline{W}_{k,1}^{(1)} \lvert \overline{X}_{k}^{(2)})-H(\overline{W}_{k,1}^{(1)} \lvert \overline{X}_{k}^{(2)},~\overline{Z}_1,~\ldots,~\overline{Z}_{k-1},~\overline{Z}_{k,2})  \nonumber \\
&\overset{(b)}{=}n_1\epsilon_1+H(\overline{W}_{k,1}^{(1)} \lvert \overline{X}_{k}^{(2)})-H(\overline{W}_{k,1}^{(1)} \lvert \overline{X}_{k}^{(2)}) = n_1\epsilon_1,
\end{align}
where $(a)$ follows from wiretap coding and $(b)$ follows since $(\overline{Z}_1,~\ldots,~\overline{Z}_{k-1},~ \overline{Z}_{k,2}) \perp ({W}_{k,1}^{(1)},~\overline{X}_k^{(2)})$.

Next consider $I_2$. We have,
\begin{align}
I_2&=I(\overline{W}_{k-1,1}^{(1)};\overline{Z}_1,~\ldots,~\overline{Z}_{k-1,1},~\overline{Z}_{k-1,2},~ \overline{Z}_k \lvert \overline{X}_k^{(2)},~\overline{W}_{k,1}^{(1)})  \nonumber \\
&=I(\overline{W}_{k-1,1}^{(1)};\overline{Z}_{k-1,1} \lvert \overline{X}_k^{(2)},~\overline{W}_{k,1}^{(1)})  +I(\overline{W}_{k-1,1}^{(1)};(\overline{Z}_1,~\ldots,~\overline{Z}_{k})\backslash \overline{Z}_{k-1,1} \lvert \overline{X}_k^{(2)},~\overline{W}_{k,1}^{(1)},~\overline{Z}_{k-1})  \nonumber \\
&=H(\overline{W}_{k-1,1}^{(1)}\lvert \overline{X}_k^{(2)},~\overline{W}_{k,1}^{(1)})  -H(\overline{W}_{k-1,1}^{(1)}\lvert \overline{X}_k^{(2)},~\overline{W}_{k,1}^{(1)},~\overline{Z}_{k-1,1}) \nonumber \\ &+I(\overline{W}_{k-1,1}^{(1)};(\overline{Z}_1,~\ldots,~\overline{Z}_{k})\backslash \overline{Z}_{k-1,1} \lvert \overline{X}_k^{(2)},~\overline{W}_{k,1}^{(1)},~\overline{Z}_{k-1})  \nonumber \\
&\overset{(a)}{=}H(\overline{W}_{k-1,1}^{(1)}) -H(\overline{W}_{k-1,1}^{(1)}\lvert \overline{Z}_{k-1,1})   +I(\overline{W}_{k-1,1}^{(1)};(\overline{Z}_1,~\ldots,~\overline{Z}_{k})\backslash \overline{Z}_{k-1,1}  \lvert \overline{X}_k^{(2)},~\overline{W}_{k,1}^{(1)},~\overline{Z}_{k-1})  \nonumber \\
&=I(\overline{W}_{k-1,1}^{(1)};\overline{Z}_{k-1,1})   I(\overline{W}_{k-1,1}^{(1)};(\overline{Z}_1,~\ldots,~\overline{Z}_{k})\backslash \overline{Z}_{k-1,1}  \lvert \overline{X}_k^{(2)},~\overline{W}_{k,1}^{(1)},~\overline{Z}_{k-1})  \nonumber \\
&\overset{(b)}{\leq}n_1\epsilon_1+I(\overline{W}_{k-1,1}^{(1)};(\overline{Z}_1,\ldots,~\overline{Z}_{k})\backslash \overline{Z}_{k-1,1}\lvert   \overline{X}_k^{(2)},~\overline{W}_{k,1}^{(1)},~\overline{Z}_{k-1})  \nonumber \\
&=n_1\epsilon_1+I(\overline{W}_{k-1,1}^{(1)};\overline{Z}_1,~\ldots,~\overline{Z}_{k-1,2},~\overline{Z}_{k} \lvert \overline{X}_k^{(2)},~\overline{W}_{k,1}^{(1)},~\overline{Z}_{k-1})  \nonumber \\
&=n_1\epsilon_1+I(\overline{W}_{k-1,1}^{(1)};\overline{Z}_1\ldots,~\overline{Z}_{k-2} \lvert \overline{X}_k^{(2)},~\overline{W}_{k,1}^{(1)},~\overline{Z}_{k-1})  \nonumber \\
&+I(\overline{W}_{k-1,1}^{(1)}; \overline{Z}_k,~\overline{Z}_{k-1,2} \lvert \overline{X}_k^{(2)},\overline{W}_{k,1}^{(1)},~\overline{Z}_{k-1,1}\overline{Z}_1,~\ldots,~\overline{Z}_{k-2}) \nonumber \\
&\overset{(c)}{=}n_1\epsilon_1+0 +I(\overline{W}_{k-1,1}^{(1)}; \overline{Z}_k,~\overline{Z}_{k-1,2} \lvert \overline{X}_k^{(2)},~\overline{W}_{k,1}^{(1)},~\overline{Z}_{k-1,1}\overline{Z}_1,~\ldots,~\overline{Z}_{k-2}) \nonumber \\
&=n_1\epsilon_1+I(\overline{W}_{k-1,1}^{(1)}; \overline{Z}_{k,1} \lvert \overline{X}_k^{(2)},~\overline{W}_{k,1}^{(1)},~\overline{Z}_{k-1,1}\overline{Z}_1,~\ldots,~\overline{Z}_{k-2}) \nonumber \\
&+I(\overline{W}_{k-1,1}^{(1)}; \overline{Z}_{k,2},~\overline{Z}_{k-1,2} \lvert \overline{X}_k^{(2)},~\overline{W}_{k,1}^{(1)},~\overline{Z}_{k-1,1},~  \overline{Z}_1,~\ldots,~\overline{Z}_{k-2},~\overline{Z}_{k,1}) \nonumber \\
&=n_1\epsilon_1+H(\overline{W}_{k-1,1}^{(1)};\lvert \overline{X}_k^{(2)},~\overline{W}_{k,1}^{(1)},~\overline{Z}_{k-1,1}\overline{Z}_1,~\ldots,~\overline{Z}_{k-2}) \nonumber \\
&-H(\overline{W}_{k-1,1}^{(1)}; \lvert \overline{X}_k^{(2)},~\overline{W}_{k,1}^{(1)},~\overline{Z}_{k-1,1}\overline{Z}_1,~\ldots,~\overline{Z}_{k-2},~\overline{Z}_{k,1}) \nonumber \\
&+I(\overline{W}_{k-1,1}^{(1)}; \overline{Z}_{k,2},~\overline{Z}_{k-1,2} \lvert \overline{X}_k^{(2)},~\overline{W}_{k,1}^{(1)},~\overline{Z}_{k-1,1},~ \overline{Z}_1,~\ldots,~\overline{Z}_{k-2},~\overline{Z}_{k,1}) \nonumber \\
&\overset{(d)}{=} n_1\epsilon_1+H(\overline{W}_{k-1,1}^{(1)};\lvert \overline{Z}_{k-1,1})-H(\overline{W}_{k-1,1}^{(1)};\lvert \overline{Z}_{k-1,1})\nonumber \\
&+I(\overline{W}_{k-1,1}^{(1)}; \overline{Z}_{k,2},~\overline{Z}_{k-1,2} \lvert \overline{X}_k^{(2)},~\overline{W}_{k,1}^{(1)},~\overline{Z}_{k-1,1},~ \overline{Z}_1,~\ldots,~\overline{Z}_{k-2},~\overline{Z}_{k,1}), \label{leakage_part3} 
\end{align}
%\thinmuskip=0mu
%\medmuskip=0mu
\thickmuskip=0mu
where $(a)$ follows since $\overline{W}_{k-1,1}^{(1)} \perp (\overline{X}_k^{(2)},~\overline{W}_{k,1}^{(1)})$ and $(\overline{W}_{k-1,1}^{(1)},~\overline{Z}_{k-1}) \perp (\overline{X}_k^{(2)},~\overline{W}_{k,1}^{(1)})$, $(b)$ follows from wiretap coding, $(c)$ follows since $(\overline{W}_{k-1,1}^{(1)},~\overline{Z}_{k-1}) \perp (\overline{Z}_1,~\ldots,~\overline{Z}_{k-2},~\overline{X}_k^{(2)},~\overline{W}_{k,1}^{(1)})$, $(\overline{Z}_1,~\ldots,~\overline{Z}_{k-2})\perp(\overline{X}_k^{(2)},~\overline{W}_{k,1}^{(1)})$ and $(\overline{Z}_1,~\ldots,~\overline{Z}_{k-1}) \perp (\overline{X}_k^{(2)},~\overline{W}_{k,1}^{(1)})$ and $(d)$ follows since $(\overline{W}_{k-1,1}^{(1)},~\overline{Z}_{k-1,1}) \perp (\overline{X}_k^{(2)},~ \overline{W}_{k,1}^{(1)},~ \overline{Z}_1,~\ldots,~\overline{Z}_{k-2})$.

But,
%\begin{align}
%(\overline{Z}_{k,2},\overline{Z}_{k-1,2}) &\leftrightarrow %(\overline{W}_{k-1}^{(1)},\overline{W}_A^{(1)}) \nonumber \\
%&\leftrightarrow %(\overline{W}_{k,1}^{(1)},\overline{Z}_{k-1,1},\overline{Z}_{k,1},\overline{Z}_1,\ldots,\%overline{Z}_{k-2})    \label{markov_relation}
%\end{align}
%From (\ref{markov_relation}), we have
\begin{align}
&I(\overline{W}_{k-1,1}^{(1)}; \overline{Z}_{k,2},~\overline{Z}_{k-1,2} \lvert \overline{X}_k^{(2)},~\overline{W}_{k,1}^{(1)},~\overline{Z}_{k-1,1},~\overline{Z}_1,~\ldots,~\overline{Z}_{k-2} ,~\overline{Z}_{k,1}) \nonumber \\
&=H(\overline{W}_{k-1,1}^{(1)}\lvert \overline{X}_k^{(2)},~\overline{W}_{k,1}^{(1)},~\overline{Z}_{k-1,1},~\overline{Z}_1,~\ldots,~\overline{Z}_{k-2},~\overline{Z}_{k,1}) \nonumber \\
&-H(\overline{W}_{k-1,1}^{(1)}\lvert \overline{X}_k^{(2)},~\overline{W}_{k,1}^{(1)},~\overline{Z}_{k-1,1},~\overline{Z}_1,~\ldots,~\overline{Z}_{k-2},~\overline{Z}_{k,1},~  \overline{Z}_{k,2},~\overline{Z}_{k-1,2}) \nonumber \\
&\overset{(a)}{=}H(\overline{W}_{k-1,1}^{(1)}\lvert \overline{Z}_{k-1,1})-H(\overline{W}_{k-1,1}^{(1)}\lvert \overline{Z}_{k-1,1}) \nonumber \\
&=0
\end{align}
where $(a)$ follows, since $(\overline{W}_{k-1,1}^{(1)},~\overline{Z}_{k-1,1}) \perp (\overline{X}_k^{(2)},~\overline{W}_{k,1}^{(1)},~\overline{Z}_1,~\ldots,~\overline{Z}_{k-2},~\overline{Z}_{k,1})$ and $(\overline{W}_{k-1,1}^{(1)},~\overline{Z}_{k-1,1}) \perp (\overline{X}_k^{(2)},~\overline{W}_{k,1}^{(1)},~\overline{Z}_1,~\ldots,\overline{Z}_{k-2},~\overline{Z}_{k,1},~\overline{Z}_{k,2},~\overline{Z}_{k-1,2})$
%where $(a_1)$ follows from (\ref{markov_relation}), $(a_2)$ follows since %$(\overline{W}_{k-1,1}^{(1)}, \overline{Z}_{k-1,1}) \perp %(\overline{W}_A^{(1)},\overline{Z}_{k,2},\overline{Z}_{k-1,2},\overline{X}_{k}^{(2)})$ %and $(a_3)$ follows since $\overline{W}_A \perp %(\overline{Z}_{k,2},\overline{Z}_{k-1,2},\overline{X}_k^{(2)})$.
Hence we have 
\begin{align}
I_2 \leq n_1\epsilon_1.
\end{align}

One can similarly prove that $I_i \leq n_1\epsilon_1$ for $i=3,4,~\ldots,~N_1$. Hence,
\begin{align}
I(\overline{W}_{k,1}^{(1)},~\overline{W}_{k-1,1}^{(1)},~\ldots,~ \overline{W}_{k-N_1,1}^{(1)}; \overline{Z}_1,~\ldots,~\overline{Z}_k \lvert &\overline{X}_k^{(2)})  \leq N_1n \epsilon_1=n_1\epsilon. \label{leakage_term_1}
\end{align}
\end{proof}
%%%%%%%%%%%%%%%%%%%%%%%%%%%%%%%%%%%%%%%%%%%%%%%end of lemma 1%%%%%%%%%%
\begin{lem}
The following inequality is satisfied
\begin{align}
I&(\overline{W}^{(1)}_{k,2},~ \overline{W}^{(1)}_{k-1,2},~ \ldots,~ \overline{W}^{(1)}_{k-N_1,2};\overline{Z}_1,~\ldots,~\overline{Z}_k \lvert \overline{X}_k^{(2)},~  \overline{W}^{(1)}_{k,1},~ \ldots,~ \overline{W}^{(1)}_{k-N_1,1})\leq 6n_1\epsilon.
\end{align}
\label{lemma_a2}
\end{lem}
\begin{proof}
	\begin{align}
&I(\overline{W}_{k,2}^{(1)},~\overline{W}_{k-1,2}^{(1)},~\ldots,~ \overline{W}_{k-N_1,2}^{(1)}; \overline{Z}_1,~\ldots,~\overline{Z}_k \lvert \overline{X}_k^{(2)},~ \overline{W}_{k,1}^{(1)},~\overline{W}_{k-1,1}^{(1)},~\ldots,~ \overline{W}_{k-N_1,1}^{(1)})  \nonumber \\ 
&=I(\overline{W}_{k,2}^{(1)},~\overline{W}_{k-1,2}^{(1)},~\ldots,~ \overline{W}_{k-N_1,2}^{(1)}; \overline{Z}_1,~\ldots,~\overline{Z}_{k-N_1-1} \lvert \overline{X}_k^{(2)},~  \overline{W}_{k,1}^{(1)},~\overline{W}_{k-1,1}^{(1)},~\ldots,~ \overline{W}_{k-N_1,1}^{(1)})  \nonumber \\
&+I(\overline{W}_{k,2}^{(1)},~\overline{W}_{k-1,2}^{(1)},~\ldots,~ \overline{W}_{k-N_1,2}^{(1)}; \overline{Z}_{k-N_1},~\ldots,~\overline{Z}_k \lvert \overline{X}_k^{(2)},~ \overline{W}_{k,1}^{(1)},~\overline{W}_{k-1,1}^{(1)},~\ldots,~ \overline{W}_{k-N_1,1}^{(1)},~\overline{Z}_1,~\ldots,~\overline{Z}_{k-N_1-1}) \nonumber \\  
&\overset{(a)}{=}0+I(\overline{W}_{k,2}^{(1)},~\overline{W}_{k-1,2}^{(1)},~\ldots,~ \overline{W}_{k-N_1,2}^{(1)}; \overline{Z}_{k-N_1},~\ldots,~\overline{Z}_k \lvert \overline{X}_k^{(2)},~\overline{W}_{k,1}^{(1)},~\overline{W}_{k-1,1}^{(1)},~\ldots,~ \overline{W}_{k-N_1,1}^{(1)},~\overline{Z}_1,~\ldots,~\overline{Z}_{k-N_1-1}) \nonumber \\
&=I(\overline{W}_{k,2}^{(1)},~\overline{W}_{k-1,2}^{(1)},\ldots,~ \overline{W}_{k-N_1,2}^{(1)}; \overline{Z}_{k-N_1,1},\ldots,~\overline{Z}_{k,1} \lvert \overline{X}_k^{(2)},~ \overline{W}_{k,1}^{(1)},~\overline{W}_{k-1,1}^{(1)},~\ldots,~ \overline{W}_{k-N_1,1}^{(1)},~\overline{Z}_1,~\ldots,~\overline{Z}_{k-N_1-1}) \nonumber \\
&+I(\overline{W}_{k,2}^{(1)},~\overline{W}_{k-1,2}^{(1)},\ldots,~ \overline{W}_{k-N_1,1}^{(1)}; \overline{Z}_{k-N_1,2},~\ldots,~\overline{Z}_{k,2} \lvert \overline{X}_k^{(2)},~ \overline{W}_{k,1}^{(1)},~\overline{W}_{k-1,1}^{(1)},~\ldots,~ \overline{W}_{k-N_1,1}^{(1)},~\overline{Z}_1,~\ldots,~\overline{Z}_{k-N_1-1},~ \nonumber \\
&~~~\overline{Z}_{k,1},~\overline{Z}_{k-1,1},~\ldots,~\overline{Z}_{k-N_1,1}) \nonumber \\
&\overset{(b)}{=}0+I(\overline{W}_{k,2}^{(1)},~\ldots,~\overline{W}^{(1)}_{k-N_1,2};\overline{Z}_{k-N_1,2},~\ldots,~\overline{Z}_{k,2}\lvert  \overline{W}_{k,1}^{(1)},~\ldots,~\overline{W}^{(1)}_{k-N_1,1},~\overline{Z}_1,~\ldots,~\overline{Z}_{k-N_1},\nonumber \\
&~~~\overline{Z}_{k-N_1,1},~\ldots,~\overline{Z}_{k-1},~\overline{X}_k^{(2)})  \nonumber \\
&\overset{(c)}{=}I(\overline{W}_{k,2}^{(1)},~\ldots,~\overline{W}^{(1)}_{k-N_1,2};\overline{Z}_{k-N_1,2},~\ldots,~\overline{Z}_{k,2}\lvert \overline{Z}_1,~\ldots,~\overline{Z}_{k-N_1},~\overline{X}_k^{(2)})  \nonumber \\
&\overset{\triangleq}{=}I(\widehat{W}_2^{(1)};\widehat{Z}_2 \lvert \widehat{Z}_1,~\overline{X}_k^{(2)}),~ \nonumber 
\end{align}
\allowdisplaybreaks
\begingroup
%\verb|\thinmuskip=0mu:| \par
\setlength{\thinmuskip}{0mu}
where $(a)$ follows, since $(\overline{W}_{k,2}^{(1)},~\ldots,~\overline{W}_{k-N_1,2}^{(1)}) \perp (\overline{Z}_1 \ldots,~ \overline{Z}_{k-N_1-1},~ \overline{W}_{k,1}^{(1)},~\ldots,~\overline{W}_{k-N_1,1}^{(1)},~\overline{X}_k^{(2)})$, $(b)$ follows, since ${(\overline{W}_{k,2}^{(1)},~\overline{W}_{k-1,2}^{(1)},~\ldots,~\overline{W}_{k-N_1,2}^{(1)})}$ is independent of the other random variables (r.v.s) in the first expression, $(c)$ follows since $(\overline{W}_{k,1}^{(1)},~\ldots,~\overline{W}_{k-N_1,1}^{(1)},~\overline{Z}_{k-N_1,1},~\ldots,~\overline{Z}_{k-1,1})$ is independent of all other r.v.s in the expression, and in the last inequality we denote the respective random sequences with their respective widehat symbols.
\endgroup

Now we observe that
\begin{align}
	I&(\widehat{W}_2^{(1)};\widehat{Z}_1,~\widehat{Z}_2 \lvert \overline{X}_k^{(2)}) \nonumber \\
	&=I(\widehat{W}_2^{(1)};\widehat{Z}_1 \lvert \overline{X}^{(2)}_k)+I(\widehat{W}_2^{(1)};\widehat{Z}_2\lvert \widehat{Z}_1,~\overline{X}_k^{(2)}) \nonumber \\
	&\overset{(a)}{=}0+I(\widehat{W}_2^{(1)};\widehat{Z}_2\lvert \widehat{Z}_1,~\overline{X}_k^{(2)}) \nonumber \\
	&\leq I(\widehat{W}_2^{(1)};\widehat{Z}_1,~\widehat{Z}_2 \lvert \overline{X}_k^{(2)}) \nonumber \\
	&=I(\widehat{W}_2^{(1)};\widehat{Z}_2 \lvert \overline{X}_k^{(2)})+I(\widehat{W}_2^{(1)};\widehat{Z}_1\lvert \widehat{Z}_2,~ \overline{X}_k^{(2)})  \nonumber \\
	&\overset{(b)}{=}0+I(\widehat{W}_2^{(1)};\widehat{Z}_1\lvert \widehat{Z}_2, \overline{X}_k^{(2)})  \label{for_proof_below}
\end{align}
where $(a)$ follows since $\widehat{W}_2^{(1)} \perp (\widehat{Z}_1,~\overline{X}_k^{(2)})$, and $(b)$ follows since $\widehat{W}_2^{(1)} \perp (\widehat{Z}_2,~ \overline{X}_k^{(2)})$.

We will also use the following notation: $\widehat{W}_1^{(1)} \triangleq (\overline{W}_{k,1}^{(1)},~ \ldots,~ \overline{W}_{k-N_1,1})$,  $A_i$ are the indices of messages transmitted in slots $1,~\ldots,~k-N_1-1$ that are used as secret keys by user $i$ for transmitting messages in slots $k-N_1,\ldots,k$,  $\overline{W}_{A_i}^{(i)}=\left(\overline{W}_k^{(i)},~ k\in A_i\right)$, $\overline{W}_{A_i^c}^{(i)}=\left(\overline{W}_k^{(i)},~ k\in \{1,~\ldots,~k-N_1-1\}\right)$, similarly we define $\overline{Z}_{A_i},~\overline{Z}_{A_i^c}$. Then we have

\begin{align}
	I&(\widehat{W}_2^{(1)}; \widehat{Z}_1 \lvert \widehat{Z}_2 ,~\overline{X}_k^{(2)}) \nonumber \\
	&\leq I(\widehat{W}_2^{(1)},~ \overline{W}_{A_1}^{(1)},~ \overline{W}_{A_2}^{(2)};\widehat{Z}_1,~ \lvert \widehat{Z}_2,~ \overline{X}_k^{(2)}) \nonumber \\
	&= I(\overline{W}_{A_1}^{(1)},~\overline{W}_{A_2}^{(2)};\widehat{Z}_1,~\lvert \overline{X}_k^{(2)},~ \widehat{Z}_2) +I(\widehat{W}_2^{(1)};\widehat{Z}_1 \lvert \overline{X}_k^{(2)},~ \widehat{Z}_2,~ \overline{W}_{A_1}^{(1)},~\overline{W}_{A_2}^{(2)}) \nonumber \\
	&\overset{(a)}{\leq}I(\overline{W}_{A_1}^{(1)},~\overline{W}_{A_2}^{(2)};\widehat{Z}_1)+I(\widehat{W}_2^{(1)};\widehat{Z}_1 \lvert \overline{X}_k^{(2)},~ \widehat{Z}_2,~ \overline{W}_{A_1}^{(1)},~\overline{W}_{A_2}^{(2)}) \nonumber \\
	&\overset{(b)}{=}I(\overline{W}_{A_1}^{(1)},~\overline{W}_{A_2}^{(2)};\widehat{Z}_1)
	+0 \nonumber \\
	&=I(\overline{W}_{A_{1},1}^{(1)},~\overline{W}_{A_{1},2}^{(1)},~\overline{W}_{A_{2},1}^{(2)},~\overline{W}_{A_{2},2}^{(2)};\widehat{Z}_1)  \nonumber \\
	&=I(\overline{W}_{A_{1},1}^{(1)},~\overline{W}_{A_{2},1}^{(2)};\widehat{Z}_1) + I(\overline{W}_{A_{1},2}^{(1)},~\overline{W}_{A_{2},2}^{(2)};\widehat{Z}_1 \lvert \overline{W}_{A_{1},1}^{(1)},~\overline{W}_{A_{2},1}^{(2)} ) \nonumber \\
	&\overset{(c)}{=}I(\overline{W}_{A_{1},1}^{(1)},~\overline{W}_{A_{2},1}^{(2)};\widehat{Z}_1)+0 \nonumber \\
	&=I(\overline{W}_{A_{1},1}^{(1)};\widehat{Z}_1)+I(\overline{W}_{A_{2},1}^{(2)};\widehat{Z}_1\lvert \overline{W}_{A_{1},1}^{(1)}) \nonumber \\
	&\leq I(\overline{W}_{A_{1},1}^{(1)},~\overline{W}_{A_{1},1}^{(2)};\widehat{Z}_1)+I(\overline{W}_{A_{2},1}^{(2)};\widehat{Z}_1\lvert \overline{W}_{A_{1},1}^{(1)}) \nonumber \\
	&\overset{(d)}{\leq} 2n_1\epsilon + I(\overline{W}_{A_{2},1}^{(2)};\widehat{Z}_1\lvert \overline{W}_{A_{1},1}^{(1)}) \nonumber \\
	&\overset{(e)}{=}2n_1\epsilon+I(\overline{W}_{A_{2},1}^{(2)};\overline{Z}_{A_2},~\overline{Z}_{A_2^c} \lvert \overline{W}_{A_{1},1}^{(1)}) \nonumber \\
	&=2n_1\epsilon+I(\overline{W}_{A_{2},1}^{(2)};\overline{Z}_{A_2} \lvert  \overline{W}_{A_{1},1}^{(1)}) +I(\overline{W}_{A_{2},1}^{(2)};\overline{Z}_{A_2^c} \lvert  \overline{W}_{A_{1},1}^{(1)},~ \overline{Z}_{A_2}) \nonumber \\
	&\overset{\triangleq}{=}2n_1\epsilon+I_1+I_2 
\end{align}
where
\begin{itemize}
	\item $(a)$ follows because $\widehat{Z}_1 \leftrightarrow (\overline{W}^{(1)}_{A_1},~\overline{W}^{(2)}_{A_2}) \leftrightarrow (\widehat{Z}_2,~\overline{X}_k^{(2)})$
	\item $(b)$ follows since $\widehat{W}_2^{(1)} \leftrightarrow (\overline{W}_{A_1}^{(1)},~\overline{W}_{A_2}^{(2)},~\widehat{Z}_2,~\overline{X}_k^{(2)}) \leftrightarrow \widehat{Z}_1$
	\item $(c)$ follows since $(\overline{W}_{A_1,2}^{(1)},~\overline{W}_{A_2,2}^{(2)}) \perp (\widehat{Z}_1,~\overline{W}_{A_1,1}^{(1)},~\overline{W}_{A_2,1}^{(2)})$
	\item $(d)$,$(j)$ and $(m)$ follows by wiretap coding
	\item $(e)$ follows since $\widehat{Z}_1 = (\overline{Z}_1,~\ldots,~\overline{Z}_{k-N_1})=(\overline{Z}_{A_2},~\overline{Z}_{A_2^c})$
\end{itemize}

Now we evaluate $I_2$,
\begin{align}
	I_2&=I(\overline{W}_{A_{2},1}^{(2)};\overline{Z}_{A_2^c} \lvert  \overline{W}_{A_{1},1}^{(1)},~ \overline{Z}_{A_2}) \nonumber \\
	&=H(\overline{W}_{A_{2},1}^{(2)} \lvert  \overline{W}_{A_{1},1}^{(1)},~ \overline{Z}_{A_2}) -H(\overline{W}_{A_{2},1}^{(2)} \lvert  \overline{W}_{A_{1},1}^{(1)},~ \overline{Z}_{A_2},~\overline{Z}_{A_2^c}) \nonumber \\
	&\overset{(a)}{=}H(\overline{W}_{A_{2},1}^{(2)} \lvert  \overline{W}_{A_{1},1}^{(1)},~ \overline{Z}_{A_2,1},~\overline{Z}_{A_2,2}) -H(\overline{W}_{A_{2},1}^{(2)} \lvert  \overline{W}_{A_{1}\cap A_2,1}^{(1)},~ \overline{Z}_{A_2,1}) \nonumber \\
	&\overset{(b)}{=}H(\overline{W}_{A_{2},1}^{(2)} \lvert  \overline{W}_{A_{1}\cap A_2,1}^{(1)},~ \overline{Z}_{A_2,1}) -H(\overline{W}_{A_{2},1}^{(2)} \lvert  \overline{W}_{A_{1}\cap A_2,1}^{(1)},~ \overline{Z}_{A_2,1}) =0
\end{align}
where $(a)$ and $(b)$ follow because $\overline{W}_{A_1,1}^{(1)}$ and $\overline{W}_{A_1,1}^{(1)}$ are used as keys only in slots $k-N_1,\ldots,~k$.

Next we evaluate $I_1$,
\begin{align}
	I_1=I&(\overline{W}_{A_{2},1}^{(2)};\overline{Z}_{A_2} \lvert  \overline{W}_{A_{1},1}^{(1)}) \nonumber \\
	&=I(\overline{W}_{A_{2}\cap A_1,1}^{(2)},~\overline{W}_{A_{2}\cap A_1^c,1}^{(2)};\overline{Z}_{A_2} \lvert  \overline{W}_{A_{1},1}^{(1)}) \nonumber \\
	&=I(\overline{W}_{A_{2}\cap A_1^c,1}^{(2)};\overline{Z}_{A_2} \lvert  \overline{W}_{A_{1},1}^{(1)})+ +I(\overline{W}_{A_{2}\cap A_1,1}^{(2)};\overline{Z}_{A_2} \lvert  \overline{W}_{A_{1},1}^{(1)},~ \overline{W}_{A_{2}\cap A_1^c,~1}^{(2)}) \nonumber \\
	&\overset{\triangleq}{=}I_3+I_4  \label{first_term}
\end{align}
Now 
\begin{align}
	I_3&=I(\overline{W}_{A_{2}\cap A_1^c,1}^{(2)};\overline{Z}_{A_2} \lvert  \overline{W}_{A_{1},1}^{(1)}) \nonumber \\ 
	&=I(\overline{W}_{A_{2}\cap A_1^c,1}^{(2)};\overline{Z}_{A_2\cap A_1},\overline{Z}_{A_2\cap A_1^c} \lvert  \overline{W}_{A_{1},1}^{(1)}) \nonumber \\
	&=I(\overline{W}_{A_{2}\cap A_1^c,1}^{(2)};\overline{Z}_{A_2\cap A_1^c} \lvert  \overline{W}_{A_{1},1}^{(1)}) +I(\overline{W}_{A_{2}\cap A_1^c,1}^{(2)};\overline{Z}_{A_2\cap A_1} \lvert  \overline{W}_{A_{1},1}^{(1)},\overline{Z}_{A_2\cap A_1^c}) \nonumber \\
	&\overset{\triangleq}{=}I_{31}+I_{32}.   \label{second_term}
\end{align}
Consider,~
\begin{align}
	I_{31}&=I(\overline{W}_{A_{2}\cap A_1^c,~1}^{(2)};\overline{Z}_{A_2\cap A_1^c} \lvert  \overline{W}_{A_{1},~1}^{(1)}) \nonumber \\
	&=I(\overline{W}_{A_{2}\cap A_1^c,~1}^{(2)};\overline{Z}_{A_2\cap A_1^c,~1},~\overline{Z}_{A_2\cap A_1^c,2} \lvert  \overline{W}_{A_{1},1}^{(1)}) \nonumber \\
	&=I(\overline{W}_{A_{2}\cap A_1^c,1}^{(2)};\overline{Z}_{A_2\cap A_1^c,1} \lvert  \overline{W}_{A_{1},1}^{(1)}) +I(\overline{W}_{A_{2}\cap A_1^c,1}^{(2)};\overline{Z}_{A_2\cap A_1^c,2} \lvert  \overline{W}_{A_{1},1}^{(1)},~\overline{Z}_{A_2\cap A_1^c,1}) \nonumber \\
	&\overset{(a)}{\leq}I(\overline{W}_{A_{2}\cap A_1^c,1}^{(1)},~\overline{W}_{A_{2}\cap A_1^c,1}^{(2)};\overline{Z}_{A_2\cap A_1^c,1}) +0\nonumber \\
	&\overset{(b)}{\leq} 2n_1\epsilon,
\end{align}
where $(a)$ follows since $\overline{Z}_{A_2\cap A_1^c,2} \perp (\overline{W}_{A_{2}\cap A_1^c,1}^{(2)},~  \overline{W}_{A_{1},1}^{(1)},~\overline{Z}_{A_2\cap A_1^c,1})$, $(b)$ follows from wiretap coding and that $\overline{W}_{A_1,1}^{(1)} \perp (\overline{W}_{A_2\cap A_1^c,1}^{(2)},~ \overline{Z}_{A_2\cap A_1^c,1})$. Next consider the  $2^{nd}$ term of (\ref{second_term}). We get
\begin{align}
	I_{32}&=I(\overline{W}_{A_{2}\cap A_1^c,1}^{(2)};\overline{Z}_{A_2\cap A_1} \lvert  \overline{W}_{A_{1},1}^{(1)},~\overline{Z}_{A_2\cap A_1^c}) \nonumber \\
	&=I(\overline{W}_{A_{2}\cap A_1^c,1}^{(2)};\overline{Z}_{A_2\cap A_1,1},~\overline{Z}_{A_2\cap A_1,2} \lvert  \overline{W}_{A_{1},1}^{(1)},~\overline{Z}_{A_2\cap A_1^c}) \nonumber \\
	&=I(\overline{W}_{A_{2}\cap A_1^c,1}^{(2)};\overline{Z}_{A_2\cap A_1,1} \lvert  \overline{W}_{A_{1},1}^{(1)},~\overline{Z}_{A_2\cap A_1^c}) +I(\overline{W}_{A_{2}\cap A_1^c,1}^{(2)};\overline{Z}_{A_2\cap A_1,2} \lvert  \overline{W}_{A_{1},1}^{(1)},~\overline{Z}_{A_2\cap A_1^c},~\overline{Z}_{A_2\cap A_1,1}) \nonumber \\
	&\overset{(a)}{=}I(\overline{W}_{A_{2}\cap A_1^c,1}^{(2)};\overline{Z}_{A_2\cap A_1,1} \lvert  \overline{W}_{A_{1},1}^{(1)},~\overline{Z}_{A_2\cap A_1^c})+0 \nonumber \\
	&=H(\overline{W}_{A_{2}\cap A_1^c,1}^{(2)} \lvert  \overline{W}_{A_{1},1}^{(1)},~\overline{Z}_{A_2\cap A_1^c}) -H(\overline{W}_{A_{2}\cap A_1^c,1}^{(2)}; \lvert  \overline{W}_{A_{1},1}^{(1)},~\overline{Z}_{A_2\cap A_1^c},~\overline{Z}_{A_2\cap A_1,1}) \nonumber \\
	&\overset{(b)}{=}H(\overline{W}_{A_{2}\cap A_1^c,1}^{(2)} \lvert  \overline{Z}_{A_2\cap A_1^c})-H(\overline{W}_{A_{2}\cap A_1^c,1}^{(2)} \lvert  \overline{Z}_{A_2\cap A_1^c}) \nonumber \\
	&=0
\end{align}
where $(a)$ follows since $\overline{Z}_{A_2\cap A_1,2} \perp (\overline{W}_{A_{2}\cap A_1^c,1}^{(2)},~\overline{W}_{A_{1},1}^{(1)},~\overline{Z}_{A_2\cap A_1^c},~\overline{Z}_{A_2\cap A_1,1})$, $(b)$ follows since $\overline{W}_{A_{1},1}^{(1)}\perp  (\overline{W}_{A_{2}\cap A_1^c,1}^{(2)},~\overline{Z}_{A_2\cap A_1^c})$ and $(\overline{W}_{A_{1},1}^{(1)},~\overline{Z}_{A_2\cap A_1,1})\perp (\overline{W}_{A_{2}\cap A_1^c,1}^{(2)},~\overline{Z}_{A_2\cap A_1^c})$.

Finally we consider
\begin{align}
	I_4&=I(\overline{W}_{A_{2}\cap A_1,1}^{(2)};\overline{Z}_{A_2} \lvert  \overline{W}_{A_{1},1}^{(1)},~ \overline{W}_{A_{2}\cap A_1^c,1}^{(2)}) \nonumber \\
	&=I(\overline{W}_{A_{2}\cap A_1,1}^{(2)};\overline{Z}_{A_2,1},~\overline{Z}_{A_2,2} \lvert  \overline{W}_{A_{1},1}^{(1)},~ \overline{W}_{A_{2}\cap A_1^c,1}^{(2)}) \nonumber \\
	&=I(\overline{W}_{A_{2}\cap A_1,1}^{(2)};\overline{Z}_{A_2,1} \lvert  \overline{W}_{A_{1},1}^{(1)},~ \overline{W}_{A_{2}\cap A_1^c,1}^{(2)}) +I(\overline{W}_{A_{2}\cap A_1,1}^{(2)};\overline{Z}_{A_2,2} \lvert  \overline{W}_{A_{1},1}^{(1)},~ \overline{W}_{A_{2}\cap A_1^c,1}^{(2)},~\overline{Z}_{A_2,1}) \nonumber \\
	&\overset{(a)}{=}I(\overline{W}_{A_{2}\cap A_1,1}^{(2)};\overline{Z}_{A_2,1} \lvert  \overline{W}_{A_{1},1}^{(1)},~ \overline{W}_{A_{2}\cap A_1^c,1}^{(2)})+0 \nonumber \\
	&=I(\overline{W}_{A_{2}\cap A_1,1}^{(2)};\overline{Z}_{A_2 \cap A_1,1},~\overline{Z}_{A_2 \cap A_1^c,1} \lvert  \overline{W}_{A_{1},1}^{(1)},~ \overline{W}_{A_{2}\cap A_1^c,1}^{(2)}) \nonumber \\
	&=I(\overline{W}_{A_{2}\cap A_1,1}^{(2)};\overline{Z}_{A_2 \cap A_1,1} \lvert  \overline{W}_{A_{1},1}^{(1)},~ \overline{W}_{A_{2}\cap A_1^c,1}^{(2)}) +I(\overline{W}_{A_{2}\cap A_1,1}^{(2)};\overline{Z}_{A_2 \cap A_1^c,1} \lvert  \overline{W}_{A_{1},1}^{(1)},~ \overline{W}_{A_{2}\cap A_1^c,1}^{(2)},~\overline{Z}_{A_2 \cap A_1,1}) \nonumber \\
	&\leq I(\overline{W}_{A_{2}\cap A_1,1}^{(2)},~\overline{W}_{A_{2}\cap A_1,1}^{(2)};\overline{Z}_{A_2 \cap A_1,1} \lvert  \overline{W}_{A_{1},1}^{(1)}) +H(\overline{Z}_{A_2 \cap A_1^c,1} \lvert  \overline{W}_{A_{1},1}^{(1)},~ \overline{W}_{A_{2}\cap A_1^c,1}^{(2)},~\overline{Z}_{A_2 \cap A_1,1}) \nonumber \\
	&-H(\overline{Z}_{A_2 \cap A_1^c,1} \lvert  \overline{W}_{A_{1},1}^{(1)},~ \overline{W}_{A_{2}\cap A_1^c,1}^{(2)},~\overline{Z}_{A_2 \cap A_1,1},~\overline{W}_{A_{2}\cap A_1,1}^{(2)}) \nonumber \\
	&\overset{(b)}{\leq}2n_1\epsilon+H(\overline{Z}_{A_2 \cap A_1^c,1} \lvert   \overline{W}_{A_{2}\cap A_1^c,1}^{(2)})-H(\overline{Z}_{A_2 \cap A_1^c,1} \lvert   \overline{W}_{A_{2}\cap A_1^c,1}^{(2)}) \nonumber \\
	&=2n_1\epsilon,
\end{align}
where $(a)$ follows, since $\overline{Z}_{A_2,2}$ is independent of the rest of the terms in the expression, $(b)$ follows because $(\overline{Z}_{A_2 \cap A_1^c,1},~\overline{W}_{A_{2}\cap A_1^c,1}^{(2)}) \perp (\overline{W}_{A_{1},1}^{(1)}, \overline{Z}_{A_2 \cap A_1,1})$ and $(\overline{Z}_{A_2 \cap A_1^c,1},~\overline{W}_{A_{2}\cap A_1^c,1}^{(2)}) \perp (\overline{W}_{A_{1},1}^{(1)},~\overline{Z}_{A_2 \cap A_1,1},~\overline{W}_{A_{2}\cap A_1,1}^{(2)})$.

Hence we have from (\ref{first_term}) that $I \leq 6n_1\epsilon$. Thus we get,

%\begin{align}
%I&(\widehat{W}_2^{(1)}; \widehat{Z}_1 \lvert \widehat{Z}_2 ,\overline{X}_k^{(2)}) %\nonumber \\ 
%&\overset{(a)}{=}I(\widehat{W}_2^{(1)}; \widehat{Z}_2 \lvert \widehat{Z}_1 %,\overline{X}_k^{(2)}) \nonumber \\
%&\leq I(\widehat{W}_1^{(1)}, %\overline{W}_{k-N_1-1,1}^{(1)},\ldots,\overline{W}_{1,1}^{(1)}, \nonumber \\
%&\qquad\qquad\qquad\widehat{W}_2^{(1)},\overline{W}_{k-N_1-1,2}^{(1)},\ldots,\overline{W}%_{1,2}^{(1)};\widehat{Z}_1, \lvert \widehat{Z}_2, \overline{X}_k^{(2)}) \nonumber \\
%&\overset{\triangleq}{=} I(\widetilde{X};\widetilde{Y} \lvert \widetilde{Z}) \nonumber \\
%&\overset{(b)}{\leq}I(\widetilde{X};\widetilde{Y}) \nonumber \\
%&=I(\overline{W}_{1,1})
%\end{align}
%where $(a)$ follows from (\ref{for_proof_below}), $(b)$ follows from (\ref{markov_2}), $(c)$ follows since $\widehat{W}_1^{(1)} \perp (\overline{X}_k^{(2)},\widehat{Z}_1)$ and $\widehat{W}_2^{(1)} \perp (\overline{X}_k^{(2)}, \widehat{W}_1^{(1)},\widehat{Z}_1)$. Hence
\begin{eqnarray}
	I(\widehat{W}_2^{(1)};\widehat{Z}_2 \lvert \widehat{Z}_1,~\overline{X}_k^{(2)}) \leq 6n_1\epsilon, \nonumber \\ 
\end{eqnarray}
whence the lemma is established.
\end{proof}

\bibliographystyle{IEEEtran}
\bibliography{single_SHAHID_IEEE_TVT_MAC_WT_JAN_12_2017}
\end{document}